\documentclass[final,1p,times]{elsarticle}

\usepackage{graphicx}%
\usepackage{multirow}%
\usepackage{booktabs}\let\cline\cmidrule
\usepackage[table,xcdraw]{xcolor}
\usepackage{amsmath,amsthm, amsfonts,amssymb,txfonts}
\usepackage{mathtools}
\usepackage{mathrsfs}%
\usepackage[title]{appendix}%
\usepackage{comment}
\usepackage{xcolor}%
\usepackage{textcomp}%
\usepackage{todonotes}
\usepackage{manyfoot}%
\usepackage{algorithm}
\usepackage{multirow}
\usepackage{pdflscape}
\usepackage{rotating}

\usepackage{tabularx}
\usepackage[noend]{algorithmic}

\usepackage{listings}%
\usepackage{graphicx}   
\usepackage{wrapfig}    
\usepackage{float}      
\usepackage{rotating}
\usepackage{url}
\usepackage{siunitx}    
\usepackage{subcaption}
\usepackage[shortcuts,acronym]{glossaries}

\usepackage{lscape}

\usepackage{tikz}
\usetikzlibrary{shapes,arrows, patterns}
\usetikzlibrary{fadings}
\usetikzlibrary{calc, intersections}
\usetikzlibrary{angles}

\newcommand{\hide}[1]{}

\newtheorem{theorem}{Theorem}[section]

\newtheorem{proposition}[theorem]{Proposition}

\newcommand{\ub}[1]{u^{#1}}         
\newcommand{\tri}[1]{\Delta^{{#1}}} 
\newcommand{\sgr}[1]{{\mathcal G}^{#1}} 
\newcommand{\gr}[2]{{\mathcal G}^{{#1}{#2}}} 

\newcommand{\gri}{\mathcal G} 

\newacronym{moco}{MOCO}{Multi-Objective Combinatorial Optimization}
\newacronym{boco}{BOCO}{Bi-Objective Combinatorial Optimization}
\newacronym{momst}{MOMST}{Multi-Objective Minimum Spanning Tree}
\newacronym{bomst}{BOMST}{Bi-Objective Minimum Spanning Tree}
\newacronym{dm}{DM}{Decision Maker}
\newacronym{moop}{MOOP}{Multi-Objective Optimization Problem}
\newacronym{mst}{MST}{Minimum Spanning Tree}

\begin{document}

\begin{frontmatter}

\title{Grouping Strategies on Two-Phase Methods for Bi-objective Combinatorial Optimization}

\author[1]{Felipe O.~Mota\corref{cor1}}\ead{felipemota@dei.uc.pt}
\author[1]{Luís Paquete}\ead{paquete@dei.uc.pt}
\author[2]{Daniel Vanderpooten}\ead{daniel.vanderpooten@lamsade.dauphine.fr}

\cortext[cor1]{Corresponding author}

\affiliation[1]{organization={University of Coimbra, CISUC/LASI – Centre for Informatics and Systems of the University of Coimbra}, 
    adressline={Rua Sílvio Lima,  Pinhal de Marrocos}, 
    postcode={3030-290}, 
    city={Coimbra}, 
    country={Portugal}}

\affiliation[2]{organization={LAMSADE, Université Paris Dauphine, Université PSL, CNRS}, 
    adressline={Pl. du Maréchal de Lattre de Tassigny}, 
    postcode={75016}, 
    city={Paris},
    country={France}}

\begin{abstract}
Two-phase methods are commonly used to solve bi-objective combinatorial optimization problems. In the first phase, all extreme supported nondominated points are generated through a dichotomic search. This phase also allows the identification of search zones that may contain other nondominated points. The second phase focuses on exploring these search zones to locate the remaining points, which typically accounts for most of the computational cost. Ranking algorithms are frequently employed to explore each zone individually, but this approach leads to redundancies, causing multiple visits to the same solutions. To mitigate these redundancies, we propose several strategies that group adjacent zones, allowing a single run of the ranking algorithm for the entire group. Additionally, we explore an implicit grouping approach based on a new concept of coverage. Our experiments on the Bi-Objective Spanning Tree Problem demonstrate the beneficial impact of these grouping strategies when combined with coverage.

\end{abstract}


\end{frontmatter}

\section{Introduction}
\label{sec:intro}

Multi-Objective Optimization is the field of study concerned with solving optimization problems with two or more conflicting objectives, called Multi-objective Optimization Problems. It has applications such as in politics [\citenum{gunasekara2014multi}], mechanics [\citenum{deb2012hybrid, jena2013multi}], economics [\citenum{mardle2000investigation}], and finance [\citenum{horn1994niched, ruspini1999automated}]. If solutions have a certain combinatorial structure (e.g. permutation, arrangement), we are facing a \acrfull*{moco} problem. 

In Multi-Objetive Optimization problems, it is very often assumed that the \acrfull*{dm} cannot express, in advance, his preferences concerning the relative importance of the  objectives. In these cases, providing the \acrshort*{dm} with a wide range of efficient solutions is important. Under the notion of Pareto optimality, a feasible solution is \emph{efficient} if there exists no other feasible solution that provides better or equal values in all objectives, with at least one strict inequality. The image of an efficient solution is a \emph{nondominated} point in the objective space. The goal is to find the set of all the efficient solutions, called the \emph{efficient set}, and/or its image in the objective space, called the \emph{nondominated set}. It is common for this type of problem to have numerous efficient solutions, but it is expected that the \acrshort*{dm} can choose an option from the efficient set or the nondominated set by inspection. Finding the complete set of nondominated points (and respective efficient solutions) for a \acrshort*{moco} problem usually requires heavy computational effort. As a consequence, many different techniques were developed to improve exact approaches for these problems.

This work focuses on exact algorithms for \acrfull*{boco} problems, in particular, on two-phase methods [\citenum{ulungu1995two}]. In this strategy, the first phase finds a subset of the nondominated set by solving a sequence of scalarized problems obtained by reformulating the original \acrshort*{boco} problem into a single-objective weighted-sum problem. 
The obtained \emph{supported} points
define \emph{search zones}
that are delimited by adjacent supported
points and by a given upper bound -- 
thus forming triangles in the 
bi-objective case --
and that may contain nondominated points. 
The goal of the second phase is to search nondominated points inside these triangles.

Several enumeration methods have been used to find nondominated points within those triangles. A widespread approach is using \emph{ranking strategy}, which enumerates solutions in an order defined by a given weighted sum objective function until it finds the remaining nondominated points. 
This two-phase method has
being acknowledged as the best alternative for the multi-objective assignment problem \cite{pedersen2008bicriterion, przybylski2008two, przybylski2010three} and the minimum spanning tree problem \cite{steiner2008computing}. 
However, this strategy has the disadvantage of finding points outside the triangle currently explored, as well as dominated points. 
Moreover, as the ranking strategy is run individually in each triangle, the same point may be found several times. As a consequence, this approach leads to a waste of computational effort. 

Our goal is to improve two-phase methods that use ranking strategies by grouping two or more adjacent triangles that will be explored together. This implies that fewer runs of the ranking algorithm are necessary, reducing redundancy, and we expect the efficiency of the second phase to be improved. However, grouping needs to be done
carefully, as grouping a large number of triangles can also lead to an unnecessarily long run of the ranking algorithm. 

We present several grouping strategies and discuss the 
trade-offs that can be obtained in terms of redundancy and search depth and illustrate the application of these techniques to the 
\acrfull*{bomst}. Assuming complete information about the nondominated set, we show how to obtain an optimal grouping by solving a shortest path problem on a complete acyclic graph where the vertices correspond to the supported points and the cost of each arc is the computational effort to explore a group of consecutive triangles.
Moreover, our experimental results on the \acrshort*{bomst} problem show that grouping pairs or triples of adjacent triangles can substantially reduce the computational effort compared to a baseline approach without any grouping.
We also show that competitive results can be achieved by exploring certain geometrical properties of those triangles. Finally, we propose an improvement for the covering technique proposed in \cite{Steiner2003},
which allows skipping some triangles from the two-phase exploration.

The remainder of this paper is organized as follows. Section~\ref{sec:defs} presents the basic concepts used in the paper. Section~\ref{sec:related} contains the state of the art of two-phase methods for \acrshort*{moco} problems. In Section~\ref{sec:group} we present how two triangles can be grouped. Section~\ref{sec:optGrouping} shows how to obtain, \emph{a posteriori}, the optimal partition of the triangles into groups. Section~\ref{sec:grouping} discusses several grouping strategies that can be employed in the second phase.
Section~\ref{sec:results} reports the experimental results for the aforementioned grouping strategies. Finally, Section~\ref{sec:conc} provides conclusions and further ideas.

\section{Definitions}
\label{sec:defs}

We introduce the following component-wise ordering in $\mathbb R^2$:
Given points $z$ and $\bar z$ in $\mathbb{R}^2$, we consider the following binary relations, respectively referred to as \emph{(Pareto) strong dominance}, \emph{weak dominance}, and \emph{dominance}:

$$
\begin{array}{lll}
    z \prec \bar z \iff & z_j < \bar z_j & \forall j \in \{1,2\}\\
    z \preceqq \bar z \iff & z_j \leq \bar z_j & \forall j \in \{1,2\}\\
    z \preceq \bar z \iff & \begin{cases}
						 z \preceqq \bar z \\
						z \neq \bar z
					 	\end{cases}  \\

\end{array}
$$

We also introduce the following non-negative cones:
$$
\begin{array}{lll}
    \mathbb R^2_\geqq & = & \left\{z \in \mathbb R^2: \mathbf 0 \preceqq  z \right\}\\
    \mathbb R^2_\geq  & = & \left\{z \in \mathbb R^2: \mathbf 0 \preceq  z \right\} =  \mathbb R^2_\geqq \setminus \{\mathbf 0 \}\\
\end{array}
$$

We assume bi-objective optimization problems 
as follows
\begin{equation}
\min_{x \in X} f(x)=\left(f_1(x), f_2(x)\right)
\label{prob:2}
\end{equation}
where $X$ is the set of feasible solutions. Let $Y=f(X)$ be the set of images of all solutions in $X$.
We assume problems with linear objective functions with integer
coefficients and integer, or in most cases, binary decision variables. Therefore, we assume in this paper 
that the objective functions take integer values, thus
$Y \subset \mathbb Z^2$. 

The weighted-sum scalarization of Problem \eqref{prob:2}
is of particular interest for our work and 
is formulated as follows, for a given weight vector
$w = (w_1, w_2) \in \mathbb R^2_\geq$
\begin{equation}
\min_{x \in X} f_w(x)=w_1 f_1(x) + w_2 f_2(x).
\label{prob:3}
\end{equation}

A point $y \in Y$ is said to be \emph{nondominated} if there is no point $y' \in Y$ such that $y' \preceq y$. Let $Y_N$ be the set of nondominated points of Problem~(\ref{prob:2}). Three types of nondominated points can be distinguished: supported extreme, supported nonextreme, and unsupported points.
The corresponding sets are denoted respectively as $Y_{NSE}$, $Y_{NSN}$ and $Y_{NU}$.
Let $\mathcal C_N$ denote the 
convex hull of $Y_N + \mathbb R^2_\geqq$,
where operator $+$ denotes the Minkowski sum.
Let $bd(\mathcal C_N)$ 
and $int(\mathcal C_N)$ 
denote the boundary and the interior 
of $\mathcal C_N$, respectively.
Points in $Y_{NSE}$ are vertices in $\mathcal C_N$, points in $Y_{NSN}$ are in $bd(\mathcal C_N) 
\setminus Y_{NSE}$, and points in $Y_{UN}$ are nondominated points in $int(\mathcal C_N)$.
We now recall the following well-known properties.
Any unsupported point of Problem~\eqref{prob:2} is not optimal for Problem~\eqref{prob:3} for any weight vector $w$. However, there always exists a weight vector $w$ for which a supported (extreme or nonextreme) point 
of Problem~\eqref{prob:2} is optimal for Problem~\eqref{prob:3}.

\section{Related work}
\label{sec:related}

The two-phase method has been widely used to solve 
bi-objective versions of many standard optimization problems, including shortest path~\cite{Mote1991,raith2009comparison}, minimum spanning tree~\cite{amorosi2022two, Andersen1996,hamacher1994spanning,steiner2008computing}, assignment~\cite{przybylski2008two,Tuyttens2000,ulungu1995two}, network flows~\cite{ehrgott1999integer,lee1993bicriteria,raith2009two,SedenoNoda2001}, knapsack~\cite{ulungu1995two,visee1998two},
set covering~\cite{prins2006two}, and  max-ordering~\cite{ehrgott2003solving}.
This method was first proposed by Ulungu 
and Teghem~\cite{ulungu1995two}. The first phase consists 
of solving a series of weighted-sum problems to 
obtain set $Y_{NSE}$ by a general 
technique known as dichotomic search 
\cite{aneja1979bicriteria}, which recursively
bisects the objective space.

In the second phase, the goal is to find
new nondominated points between consecutive
 supported extreme points found in the first
phase. 
Let $Y_{NSE} = \{y^1, y^2, \ldots, y^{m+1}\}$ be the extreme supported points ordered by increasing value of objective $f_1$ and decreasing value of objective $f_2$.
Each pair of consecutive points in $Y_{NSE}$, $y^i$ and $y^{i+1}$, and point $(y^{i+1}_1,y^i_2)$
 define a \emph{triangle} $\tri{i}$, for $i \in \{1,\ldots,m\}$. Assuming integrality of the feasible points, each triangle $\tri{i}$ is associated with a \emph{local upper bound} $\ub{*} = (y^{i+1}_1 -1,  y^i_2 - 1)$ delimiting the search zone where points in $Y_N \setminus Y_{NSE}$ might lie within triangle $\tri{i}$. More precisely, any point $y$ in $Y_N \cap \tri{i}$ is such that $y \preceqq \ub{*}$. Figure \ref{fig:ub} (left) illustrates a triangle $\tri{i}$
defined by points $y^i$ and $y^{i+1}$
and the location of upper bound $\ub{*}$.

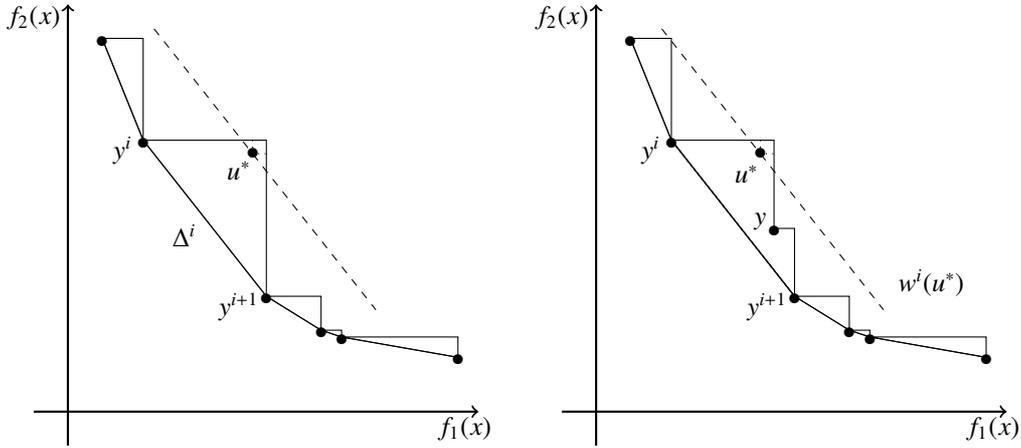
\begin{figure}[t]
\begin{subfigure}[t]{0.43\textwidth}
    \centering
            \begin{tikzpicture}[scale=0.9]
                \draw[black] 
                (1.5,6) coordinate (A3) 
                -- (2.1,4.5) coordinate (A1) 
                -- (3.9,2.2) coordinate (B2)  
                -- (4.7,1.7) coordinate (B1) 
                --  (5,1.6)  coordinate (C1) 
                --  (6.7,1.3)  coordinate (C2);
                
                \draw[black] (A3) --  (A3 -| A1)  -- (A1) --  (A1 -| B2) -- (B2) -- (B2 -| B1)  -- (B1) -- (B1 -| C1) -- (C1) -- (C1 -| C2) -- (C2) -- (C1) -- (B1) --  (B2) -- (A1) -- cycle;
                
                \draw[thick,->] (0.5,0.5) -- (7,0.5);
                \draw[thick,->] (1,0) -- (1,6.5);
                \draw ($(A3) - (0,0.05)$) node {
                \textbullet};
                \draw ($(A1) - (0,0.05)$) node {
                \textbullet};
                \draw ($(B1) - (0,0.05)$) node { \textbullet};
                \draw ($(B2) - (0,0.05)$) node { \textbullet};
                \draw ($(C1) - (0,0.05)$) node { \textbullet}; 
                \draw ($(C2) - (0,0.05)$) node { \textbullet};
                \draw[black] ($(A1 -| B2) - (0.2,0.2)$) node { \textbullet};   
                

                \node at ($(A1 -| B2) - (0.4,0.45)$) {$\ub{*}$};
                \node at ($(A1 -| B2) - (1.2,1.4)$) {$\Delta^{i}$};
                \node at ($(A1) - (0.3,0.15)$) {$y^i$};
                \node at ($(B2) - (0.4,0.15)$) {$y^{i+1}$};

                \draw[dashed] ($(A1 -| B2) - 0.8*(1.8,-2.3) - (0.2,0.2)$) -- 
                ($(A1 -| B2) + (1.8,-2.3) - (0.2,0.2)$);

                \draw[dotted, thin] ($(A1 -| B2) - (0.2,0)$) -- ($(A1 -| B2) - (0.2,0.2)$) -- ($(A1 -| B2) - (0,0.2)$);
                
                \node at (6.8, 0.25) { $f_1(x)$};
                \node at (0.5, 6.3) {$f_2(x)$};
            
            \end{tikzpicture}
\end{subfigure}
\qquad \quad
\begin{subfigure}[t]{0.43\textwidth}
    \centering
            \begin{tikzpicture}[scale=0.9]

                \draw[black] 
                (1.5,6) coordinate (A3) 
                -- (2.1,4.5) coordinate (A1) 
                -- (3.9,2.2) coordinate (B2)  
                -- (4.7,1.7) coordinate (B1) 
                --  (5,1.6)  coordinate (C1) 
                --  (6.7,1.3)  coordinate (C2);
                     );
              
                \draw[black] 
                (A1) 
                -- (3.6,4.5) coordinate (Z1a) 
                -- (3.6,3.2) coordinate (Z1) 
                -- (3.9,3.2) coordinate (Z1b) 
                -- (B2);
              
                \draw[black] (A3) --  (A3 -| A1)  -- (A1) --   (B2) -- (B2 -| B1)  -- (B1) -- (B1 -| C1) -- (C1) -- (C1 -| C2) -- (C2) -- (C1) -- (B1) --  (B2) -- (A1) -- cycle;
                
                \draw[thick,->] (0.5,0.5) -- (7,0.5);
                \draw[thick,->] (1,0) -- (1,6.5);

                \draw ($(Z1) - (0,0.05)$) node {
                \textbullet};
                \draw ($(A3) - (0,0.05)$) node {
                \textbullet};
                \draw ($(A1) - (0,0.05)$) node {
                \textbullet};
                \draw ($(B1) - (0,0.05)$) node { \textbullet};
                \draw ($(B2) - (0,0.05)$) node { \textbullet};
                \draw ($(C1) - (0,0.05)$) node { \textbullet}; 
                \draw ($(C2) - (0,0.05)$) node { \textbullet};
                \draw[black] ($(A1 -| Z1)  - (0.2,0.2)$) node { \textbullet};   
                
                \draw[dashed] ($(A1 -| Z1) - 0.8*(1.8,-2.3) - (0.2,0.2)$) -- 
                ($(A1 -| Z1) + (1.8,-2.3) - (0.2,0.2)$);

                \node at ($(A1 -| Z1) + (0.4,0.2) - (0.8,0.7)$) {$\ub{*}$};
                \node at ($(A1) - (0.3,0.15)$) {$y^i$};
                \node at ($(B2) - (0.4,0.15)$) {$y^{i+1}$};
                \node at ($(Z1) + (-0.2,0.1)$) {$y$};

                \node at ($(B2) + (2,0.2)$) {$w^i(\ub{*})$};

                \draw[dotted, thin] ($(A1 -| Z1) - (0.2,0)$) -- ($(A1 -| Z1) - (0.2,0.2)$) -- ($(A1 -| Z1) - (0,0.2)$);

                \node at (6.8, 0.25) { $f_1(x)$};
                \node at (0.5, 6.3) {$f_2(x)$};
            
            \end{tikzpicture}    
\end{subfigure}
\caption{Illustration of a triangle $\tri{i}$ 
and its upper bound $\ub{*}$ (left)
and the updated upper bound $\ub{*}$ given 
that point $y$ was found in $\tri{i}$ (right)}
\label{fig:ub}
\end{figure}

Some of the two-phase methods above use a ranking algorithm for the
second phase~\cite{amorosi2022two,ehrgott2003solving,przybylski2008two,raith2009two,steiner2008computing}.
The concept of ranking algorithms for single-objective optimization problems was proposed by Murty \cite{murty1968algorithm} for the assignment problem and has been adapted to many other applications \cite{eppstein16}. It requires that an optimal solution for the problem is known. From that solution, by fixing variables and generating a sequence of sub-problems, it generates the $k$ best solutions in non-decreasing order of the objective function value. 

While the upper bound $u^*$ provides a natural stopping criterion for the ranking algorithm, tighter bounds are obtained as new points are found within triangle $\tri{i}$.
Let $w^i = \left(w_1^i,w_2^i\right)$ be the weight vector with respect to $\tri{i}$, where $w^i_1 = y^i_2-y_2^{i+1}$ and $w^i_2 = y^{i+1}_1-y_1^{i}$. 
We denote by $w^i(z)$ the weighted sum value of a point $z \in \mathbb Z^2$ with respect to weight vector $w^i$, that is,

$$w^{i}(z) = w^{i}_1 z_1 + w^{i}_2 z_2$$

Let $\left\{y^i, y^{i_1}, y^{i_2}, \ldots, y^{i_r}, y^{i+1}\right\} \subseteq Y_N \cap \tri{i}$, such that 
$y_1^i < y_1^{i_1} < y_1^{i_2} < \cdots < y_1^{i_r} < y_1^{i+1}$. Let $U^i$ be the set of \emph{local upper bounds} of this set of points defined as follows.

\begin{equation}
\label{eq:U}
 U^i = \left\{\left(y_1^{i_1} - 1,y_2^i - 1\right), \left(y_1^{i_2} - 1,y_2^{i_1} - 1\right), \ldots,\left(y_1^{i+1} - 1,y_2^{i_r} - 1\right)\right\}
\end{equation} 
 
We redefine the upper bound of $\tri{i}$  from $U^i$ as follows
$$
\ub{*} \in \arg\max_u \left\{w^{i}(u) \mid u \in U^i \right\}
$$
Figure \ref{fig:ub} (right) shows an updated upper bound $u^*$ in $\tri{i}$ given that a point $y \in Y_N$ is located 
inside $\tri{i}$. In this case, $U^i = 
\left \{\left(y_1 - 1,y^i_2 - 1\right),\left(y^{i+1}_1 - 1,y_2 - 1\right)
\right\}$.
The dashed line indicates the level set for $w^i(u^*)$.

When exploring $\tri{i}$, the ranking algorithm may enumerate 
solutions (efficient or not) whose points are located in other 
zones. Given that those solutions are likely to be enumerated 
again, the successive application of the ranking algorithm leads to redundancy. 
The information obtained outside the currently explored triangle is oftentimes discarded after the enumeration procedure. To our knowledge, only the work by Steiner and Radzik~\cite{steiner2008computing} partially takes advantage of 
the points found
outside the current zone of interest. In the following sections, we address this wasted effort and investigate the search on multiple zones simultaneously.

\section{Groups and groupings}
\label{sec:group}

Let a sequence of consecutive triangles $\tri{i}, \tri{i+1}, \ldots, \tri{j}$ define a  \emph{group} $\gr{i}{j}$.  
In particular, when a group consists of one triangle only $\tri{i}$ it is denoted as $\sgr{i}$. We denote by $Y^{ij}_N \subseteq Y_N$ the set of nondominated points within group $\gr{i}{j}$.

In the following, we extend the previous notions with respect to a group. We now define $w^{ij} = \left(w^{ij}_1,w^{ij}_2\right)$  with respect to 
a group $\gr{i}{j}$, where
$w^{ij}_1=y^i_2 - y^{j+1}_2$ and $w^{ij}_2 = y^{j+1}_1 - y^i_1$ and we denote by $w^{ij}(z)$ the weighted sum value of a point $z \in \mathbb Z^2$ with respect to weight vector $w^{ij}$.

Let $\gri = \{ \gr{i_1}{i_2},\gr{i_2+1,}{i_3},\ldots,
\gr{i_h}{i_{h+1}}\}$ be a \emph{grouping} where 
$i_1 = 1$, $i_{h+1} = m$, and $i_\ell \leq i_{\ell+1}$, $\ell = 1,\ldots,h$,
i.e., a partition of triangles $\tri{1},\ldots,\tri{m}$ 
into $h$ consecutive groups. For each group $\gr{i}{j} \in \gri$, we assume that points are generated by a ranking algorithm in non-decreasing order of the weighted sum $w^{ij}(y)$, starting from the solution corresponding to the following point
\begin{equation}
\label{eq:y}
y^* \in \arg\min_y \left\{w^{ij}(y) \mid y \in 
Y^{ij}_N  \right\}
\end{equation}

Note that $y^*$ corresponds to one of the extreme supported nondominated points $y^i,\ldots,y^{j+1}$.

The upper bound $u^*$ of a group $\gr{i}{j}$ is now computed by taking into account the union of all local upper bounds of the triangles involved in the group, that is, 
\begin{equation}
\label{eq:u}
\ub{*} \in \arg\max_u \left\{w^{ij}(u) \mid u \in U^{ij}  \right\}.
\end{equation}
where $U^{ij} = U^i \cup \cdots \cup U^j$.
Figure~\ref{fig:2} illustrates an example of 
a group $\gr{i}{j}$, its upper bound $\ub{*}$,
starting solution $y^*$ and level sets as 
dashed lines associated to the weighted-sum values 
of different points. Similar to the ranking approaches described in the previous section, the upper bound $\ub{*}$ can be used as a termination criterion for a 
ranking algorithm that explores group $\gr{i}{j}$. 

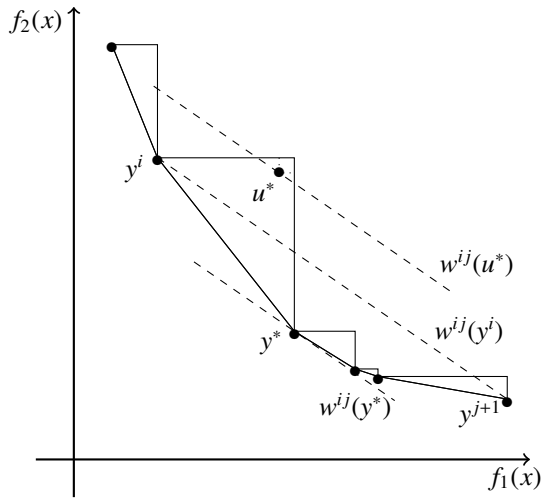
\begin{figure}[t]
    \centering
            \begin{tikzpicture}
            
                \draw[black] 
                (1.5,6) coordinate (A3) 
                -- (2.1,4.5) coordinate (A1) 
                -- (3.9,2.2) coordinate (B2)  
                -- (4.7,1.7) coordinate (B1) 
                --  (5,1.6)  coordinate (C1) 
                --  (6.7,1.3)  coordinate (C2);
                
                \draw[black] (A3) --  (A3 -| A1)  -- (A1) --  (A1 -| B2) -- (B2) -- (B2 -| B1)  -- (B1) -- (B1 -| C1) -- (C1) -- (C1 -| C2) -- (C2) -- (C1) -- (B1) --  (B2) -- (A1) -- cycle;
                
                \draw[thick,->] (0.5,0.5) -- (7,0.5);
                \draw[thick,->] (1,0) -- (1,6.5);
                \draw ($(A3) - (0,0.05)$) node {
                \textbullet};
                \draw ($(A1) - (0,0.05)$) node {
                \textbullet};
                \draw ($(B1) - (0,0.05)$) node { \textbullet};
                \draw ($(B2) - (0,0.05)$) node { \textbullet};
                \draw ($(C1) - (0,0.05)$) node { \textbullet}; 
                \draw ($(C2) - (0,0.05)$) node { \textbullet};
                \draw[black] ($(A1 -| B2) - (0.2,0.2)$) node { \textbullet};   
                
                \draw[dashed] ($(A1 -| B2) - 0.5*(3.3,-2.3)  - (0.2,0.2)$) -- 
                ($(A1 -| B2) + 0.7*(3.3,-2.3)  - (0.2,0.2)$);

                \draw[dashed] ($(B2) - 0.4*(3.3,-2.3)$) -- 
                ($(B2) + 0.4*(3.3,-2.3)$);

                \draw[dashed] (A1) -- (C2);

                \node at ($(A1 -| B2) - (0.4,0.45)$) {$\ub{*}$};
                \node at ($(A1) - (0.3,0.15)$) {$y^i$};
                \node at ($(C2) - (0.35,0.2)$) {$y^{j+1}$};
                \node at ($(B2) - (0.3,0.15)$) {$y^*$};
                \node at ($(B2) + (0.8,-1)$) {$w^{ij}(y^*)$};
                \node at ($(B2) + (2.3,0)$) {$w^{ij}(y^i)$};
                \node at ($(B2) + (2.4,0.9)$) {$w^{ij}(\ub{*})$};
                
                \draw[dotted, thin] ($(A1 -| B2) - (0.2,0)$) -- ($(A1 -| B2) - (0.2,0.2)$) -- ($(A1 -| B2) - (0,0.2)$);
                
                \node at (6.8, 0.25) { $f_1(x)$};
                \node at (0.5, 6.3) {$f_2(x)$};
            
            \end{tikzpicture}
    \caption{Illustration of group $\gr{i}{j}$, its
    upper bound ($\ub{*}$) and starting solution ($y^*$).}
    \label{fig:2}
\end{figure}

Algorithm \ref{alg:code1} shows a pseudo-code that returns the nondominated set by exploring groups in a given grouping $\gri$. The set $Y_{NSE}$ is obtained in the first phase.
At each iteration of the outer loop, a set $S^{ij}$ collects the nondominated points within group $\gr{i}{j}$. 
Once the outer loop terminates, $S^{ij}$ contains the elements of $Y^{ij}_N$ for the underlying problem to be solved. 
At the end of Algorithm \ref{alg:code1}, set $S$ contains all elements in $Y_N$.
Procedure next($y^*$,  $w^{ij}$) calls the ranking algorithm to return the next best feasible point using weight vector $w^{ij}$, and procedure updateND($y^*$, $S^{ij}$) updates $S^{ij}$ if the visited point weakly dominates $u^*$ and is not dominated by any other point in the set. It is worth noting that once a point is inserted in $S^{ij}$, it is guaranteed to be in $Y_N$.
Moreover, $y^*$ may be found more than once if it has multiple associated solutions. Finally, method updateUB($y^*$, $S^{ij}$) removes obsolete upper bounds from $U^{ij}$, inserts new, tighter ones, and returns the highest upper bound $\ub{*}$.  
The while loop in Algorithm  \ref{alg:code1} terminates when $w^{ij}(y^*)$ exceeds the value $w^{ij}(\ub{*})$. 
For efficiency reasons, both $S^{ij}$ and $U^{ij}$ are maintained by a collection of data structures, one per each triangle $\tri{\ell}$ within group $\gr{i}{j}$. 
In order to perform efficient update operations, each of these data structures can
be implemented as a balanced binary tree. 

A special type of triangle should be considered when using groups. We call $\Delta_i$ an \emph{empty triangle} when $y^{i+1}_1 - y^{i}_1 = 1$ or $y^{i}_2 - y^{i+1}_2 = 1$. The integrality assumption about the objective function values (made in Section \ref{sec:defs}) implies that there is no unsupported or supported non-extreme point in $\Delta_i$. Thus, running a ranking algorithm in this zone is unnecessary. Combined with the fact that empty triangles are more prone to be found in the uppermost (highest values for $f_2$) and rightmost (highest values for $f_1$)  parts of the objective space, we do not use supported points in such regions for the second phase if they would form an empty triangle. For any second-phase grouping strategies defined in this paper, an empty triangle $\Delta_i$ is only considered if there are two non-empty triangles $\Delta_j$ and $\Delta_k$ such that $j < i < k$.

\begin{algorithm}[t]
        \caption{Algorithm for groups} \label{alg:code1}
   \begin{algorithmic}[1]
   \REQUIRE $\gri, Y_{NSE}$
    \STATE $S \gets \emptyset$
    \FORALL{$\gr{i}{j} \in \gri$}
        \STATE $S^{ij} \gets \{y^i,y^{i+1},\ldots, y^j, y^{j+1}\}$
        \STATE $y^* \gets \arg \min \left\{w^{ij}(y) \mid y \in S^{ij}\right\}$ 
        \STATE $u^* \gets \mbox{genUB}(S^{ij})$
        \WHILE {$w^{ij}(y^*)\leq w^{ij}(\ub{*})$}
            \STATE $y^* \gets \mbox{next}(y^*, w^{ij})$ 
            \STATE $S^{ij} \gets \mbox{updateND}(y^*,S^{ij})$
            \STATE $\ub{*} \gets \mbox{updateUB}(y^*, S^{ij})$
        \ENDWHILE
        \STATE  $S \gets S \cup \left\{S^{ij}\right\}$
    \ENDFOR
    \RETURN $S$         
   \end{algorithmic}
\end{algorithm}

In the following, we discuss a particular case 
where grouping a set of adjacent triangles proves  advantageous.
 Consider a set of consecutive supported points that are collinear and that define a set of adjacent triangles. Then, it is always beneficial, in terms of computational effort,
 to apply the ranking algorithm once to the group of these triangles rather than to apply it to each triangle separately. Actually, the computational effort required for the group corresponds to the computational effort required for \emph{one} specific triangle of this group as stated in the next result, which clearly shows that we should regroup all these triangles rather than considering any partition of these triangles.

\medskip
\begin{proposition}
\label{prop:1}
Let $y^i, \ldots, y^{j+1}$ be $j-i+2$  
consecutive supported points that are collinear and that define $j-i+1$ adjacent triangles $\Delta^i,\ldots,\Delta^j$, $j-i>0$. Then the computational effort required for group $\gr{i}{j}$ coincides with the computational effort required for one of the triangles of this group.
\end{proposition}

\begin{proof}
First observe that, due to the collinearity assumption, the weights used for the successive explorations of each triangle are the same (up to a factor) and correspond to the weight $w^{ij}$ used for the unique exploration of group $\gr{i}{j}$. Consider now the exploration of each of the $j-i+1$ adjacent triangles $\Delta^i,\ldots,\Delta^j$ that lead to generate all nondominated points in each triangle. Each exploration stops when generating a point which reaches the (updated) upper bound of this triangle. 
Among these, let $y^*$ be the point with the largest weighted sum using weight $w^{ij}$ and $\Delta^*$ the triangle which was explored when generating $y^*$. Then exploring $\Delta^*$, or equivalently $\gr{i}{j}$, covers all triangles in $\Delta^i,\ldots,\Delta^j$ with the same number of enumerated solutions.
\end{proof}

The result above suggests that it
is beneficial to group adjacent triangles defined by supported points close to collinearity. 
However, grouping consecutive triangles defined by supported points that are not collinear can lead to unnecessary computational effort.
This behavior is shown in Figure~\ref{fig:badGroup} 
in a hypothetical example with three
adjacent triangles formed by four supported points,
$y_1$, $y_2$, $y_3$, and $y_4$.
The left plot illustrates the application
of the ranking algorithm to each of 
the triangles. 
The dashed lines correspond to the highest weighted-sum level for which a point ($\bar y^1$, $\bar y^2$, $\bar y^3$, 
respectively) was found in each triangle, surpassing the respective upper bound. 
The right plot shows the application
of the ranking algorithm on the group formed
by the three triangles, which terminates
at the highest of the three points,
with respect to the weighted sum explored
by the algorithm. Despite the redundant computations (gray) 
shown in the left plot, the grouping of the three triangles on the right 
plot leads to a much wider region
to be explored by the ranking 
algorithm. 

\begin{figure}[t]
\begin{subfigure}[t]{0.43\textwidth}
\centering
        \begin{tikzpicture}          
            \draw[thick,->] (0.5,0.5) -- (7,0.5);
            \draw[thick,->] (1,0) -- (1,6.5);
            \node at (6.8, 0.25) { $f_1(x)$};
            \node at (0.5, 6.3) {$f_2(x)$};

            \coordinate (P1) at (1.7,6);
            \coordinate (P2) at (2.3,4);
            \coordinate (P3) at (4.3,1.6);
            \coordinate (P4) at (6,1);
            
            \draw[gray!40, fill = gray!40] 
                (1.98, 5.07) --
                (2.3, 4.7) --
                (2.63, 3.59) --
                (P2) -- cycle;
                
            \draw[gray!40, fill = gray!40] 
                (4.1, 1.88) --
                (4.9, 1.58) --
                (5.15, 1.3) --
                (P3) -- cycle;

           \node at ($(P1 -| P2) + (-0.12,0.2)$) {$\bar y^{1}$};
	   \draw[black] ($(P1 -| P2) - (0.35,0.15)$) node { \textbullet};   

           \node at ($(P2 -| P3) - (0.52,0.52)$) {$\bar y^{2}$};
	   \draw[black] ($(P2 -| P3) - (0.78,0.78)$) node { \textbullet};   

           \node at ($(P3 -| P4) + (0.2,-0.25)$) {$\bar y^{3}$};
	   \draw[black] ($(P3 -| P4) - (0.15,0.35)$) node { \textbullet};

	        \draw[black] (P1) node {\small \textbullet}; 
           \node at ($(P1) + (-0.2,-0.25)$) {$ y^{1}$};            
	        \draw[black] (P2) node {\small \textbullet}; 
           \node at ($(P2) + (-0.2,-0.25)$) {$ y^{2}$};            
	        \draw[black] (P3) node {\small \textbullet}; 
           \node at ($(P3) + (-0.2,-0.25)$) {$ y^{3}$};            
	        \draw[black] (P4) node {\small \textbullet}; 
           \node at ($(P4) + (-0.2,-0.25)$) {$ y^{4}$};

            \draw (P1) -- (P1 -| P2) --
                  (P2) -- (P2 -| P3) -- 
                  (P3) -- (P3 -| P4) --
                  (P4) -- (P3) -- (P2) -- cycle;

            \draw[dashed] 
            ($(P2) + (0.2,0) + 0.22*(0.6, -2)$) --
            ($(P2) + (0.2,0) - 1.1*(0.6, -2)$);

            \draw[dashed]
            ($(P3) + (0,0.7) + 0.425*(2, -2.4)$) --
            ($(P3) + (0,0.7) - 1.18*(2, -2.4)$);

            \draw[dashed] 
            
            ($(P4) + (0,0.2) + 0.22*(1.7, -0.6)$) --
            ($(P4) + (0,0.2) - 1.1*(1.7, -0.6)$);
        \end{tikzpicture}
\end{subfigure}
\qquad \quad
\begin{subfigure}[t]{0.43\textwidth}
\centering
        \begin{tikzpicture}          
            \draw[thick,->] (0.5,0.5) -- (7,0.5);
            \draw[thick,->] (1,0) -- (1,6.5);
            \node at (6.8, 0.25) { $f_1(x)$};
            \node at (0.5, 6.3) {$f_2(x)$};

            \coordinate (P1) at (1.7,6);
            \coordinate (P2) at (2.3,4);
            \coordinate (P3) at (4.3,1.6);
            \coordinate (P4) at (6,1);

            \draw (P1) -- (P1 -| P2) --
                  (P2) -- (P2 -| P3) -- 
                  (P3) -- (P3 -| P4) --
                  (P4) -- (P3) -- (P2) -- cycle;

            \draw[dashed] ($(P4) + (0.15,0) + 0.002*(4.3, -5)$) --
            ($(P4) + (0.15,0) - 1.04*(4.3, -5)$);]

           \node at ($(P1 -| P2) + (-0.12,0.2)$) {$\bar y^{1}$};
	   \draw[black] ($(P1 -| P2) - (0.35,0.15)$) node { \textbullet};   

           \node at ($(P2 -| P3) - (0.59,0.52)$) {$\bar y^{2}$};
	   \draw[black] ($(P2 -| P3) - (0.78,0.78)$) node { \textbullet};   

           \node at ($(P3 -| P4) + (0.2,-0.25)$) {$\bar y^{3}$};
	   \draw[black] ($(P3 -| P4) - (0.15,0.35)$) node { \textbullet};   

    	        \draw[black] (P1) node {\small \textbullet}; 
           \node at ($(P1) + (-0.2,-0.25)$) {$ y^{1}$};            
	        \draw[black] (P2) node {\small \textbullet}; 
           \node at ($(P2) + (-0.2,-0.25)$) {$ y^{2}$};            
	        \draw[black] (P3) node {\small \textbullet}; 
           \node at ($(P3) + (-0.2,-0.25)$) {$ y^{3}$};            
	        \draw[black] (P4) node {\small \textbullet}; 
           \node at ($(P4) + (-0.2,-0.25)$) {$ y^{4}$};

        \end{tikzpicture}    
\end{subfigure}
\caption{Grouping triangles defined by supported points far from collinear may not be beneficial.}
\label{fig:badGroup}
\end{figure}
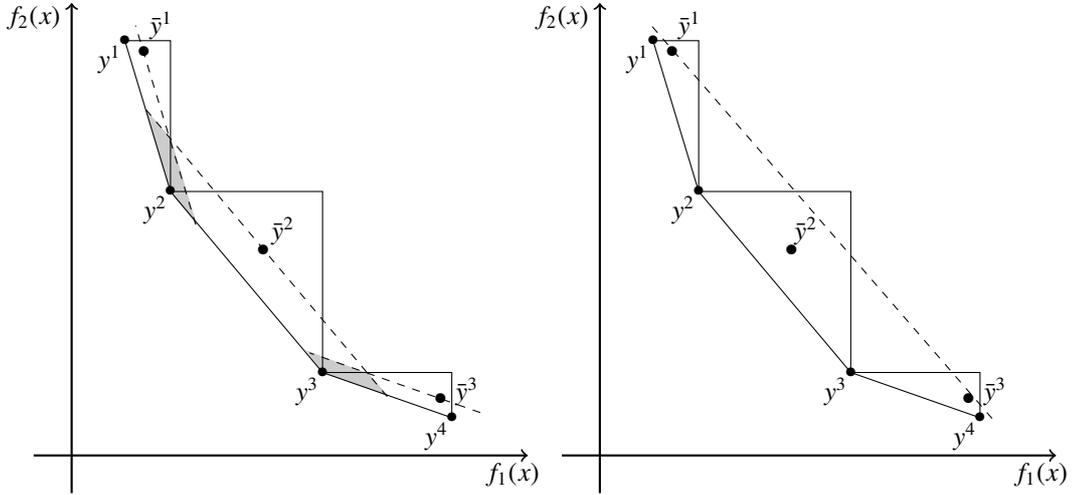

\section{Optimal grouping}
\label{sec:optGrouping}

An important aspect of this work is to develop methods for constructing an \emph{optimal} grouping that minimizes the total computational effort required for its exploration.
In this section, we show that such an optimal grouping can efficiently be constructed if the computational cost for all possible groups is known in advance.

For a given instance, we represent all possible groupings by a valued digraph $G = (V,A)$, called \emph{grouping graph}. Considering that the $m+1$ points in $Y_{NSE}$, found in the first phase, are ordered by increasing value of $f_1$, we have $V =  \{1, 2, \ldots, m+1\}$ where vertex $i$ corresponds to point $y^i \in Y_{NSE}$. Then $A = \{(i,j) \in V \times V: i<j\}$ and arc $(i,j+1) \in A$ corresponds to group $\gr{i,}{j}$ while arc $(i,i+1)$ corresponds to group $\gr{i}{}$. It follows that all feasible groupings are in one-to-one correspondence with the set of all paths from $1$ to $m+1$. 

Let $\mu^{ij}$  be the cost associated with the exploration of group $\gr{i}{j}$, 
each arc $(i,j+1)$ is then valued by $\mu^{ij}$. An optimal grouping in the sense 
of the considered cost is then obtained by computing the shortest path from $1$ to $m+1$. Observing that, by construction, $G$ is without cycles, naturally topologically ordered, and contains $\frac{(m+1)(m+2)}{2}$ arcs, the determination of an optimal path is performed very efficiently in $\mathcal{O}(m^2)$ time, 
for instance, by using the \emph{pulling} algorithm \cite{ahuja1993network}. We discuss an appropriate choice of $\mu^{ij}$ in Section \ref{sec:optRes}. 

For a given instance, this method allows us to determine the optimal grouping, offering additional insights into its structure and serving as a reference for any grouping approach.

\section{Grouping strategies}
\label{sec:grouping}

In this section, we propose heuristic strategies to form a grouping for an instance of a  \acrshort*{boco} problem. 
We first introduce \emph{grouping measures} that allow making
decisions on how to form groups (Section~\ref{sub:measures}). 
Then, we consider two main grouping strategies: i) \emph{a priori} strategies (Section~\ref{sub:apriori}), which define the grouping to be explored immediately after all the extreme supported points have been found in the first phase; 
ii) \emph{dynamic} strategies (Section~\ref{sub:dynamic}), which iteratively select the next group to be explored based on the information gathered from the current set of nondominated points found so far. 

\subsection{Grouping measures}
\label{sub:measures}

The information obtained on the location of supported points and corresponding triangles in the objective space after completing the first phase should, in principle, be used to create more effective groupings.
This section focuses on measures that allow to define groupings. We 
establish the relation of these measures with the main results 
in Section 4.

\subsubsection{Group angle}

Given a group $\gr{i}{j}$, its \emph{angle} $\theta^{ij}$ corresponds to the largest angle formed at the intersection of the lines passing through the line segments $\overline{y^iy^{i+1}}$ and $\overline{y^{j}y^{j+1}}$. Figure~\ref{fig:angleMeasure} shows an example of this measure in the special case of a group of size two (left plot) and in the general case (right plot). 

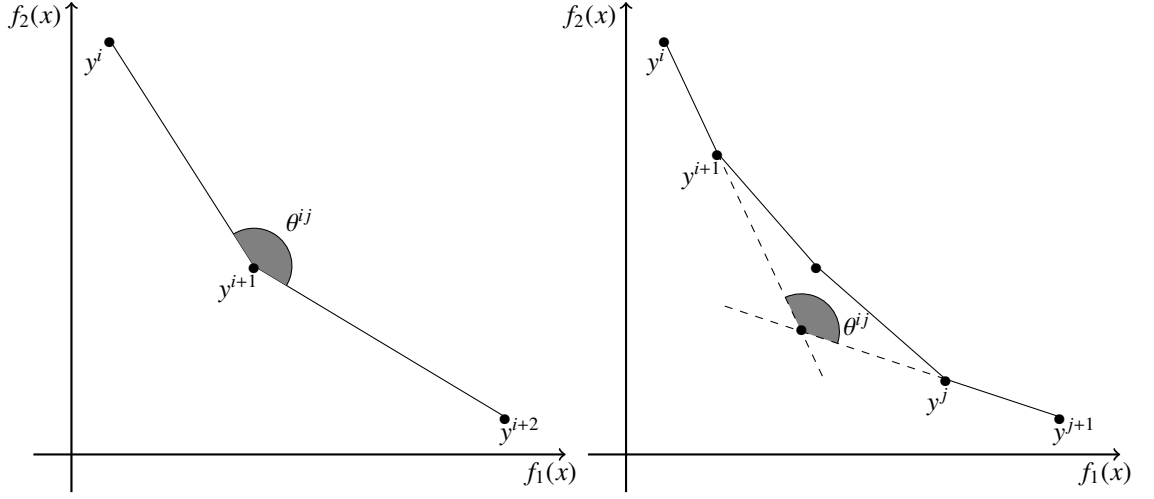
\begin{figure}[t]
    \centering
    \begin{subfigure}[t]{0.43\textwidth}
    \begin{tikzpicture}

    \draw (1,6) coordinate (A3)  -- (2.9,3) coordinate (B2)  --   (6.2,1)  coordinate (C1);

    \pic [draw, fill = gray, opacity = 0.5] {angle = C1--B2--A3};

    \draw ($(B2) + (0.6,0.6)$) node {$\theta^{ij}$};
    
    \draw[thick,->] (0,0.5) -- (7,0.5);
    \draw[thick,->] (0.5,0) -- (0.5,6.5);
    \draw ($(A3) - (0,0.05)$) node { \textbullet};
    \draw ($(B2) - (0,0.05)$) node { \textbullet};
    \draw ($(C1) - (0,0.05)$) node { \textbullet};

    \draw ($(A3) + (-0.2, -0.3)$) node { $y^{i}$};
    \draw ($(B2) + (-0.2, -0.3)$) node { $y^{i+1}$};
    \draw ($(C1) + (0.2, -0.2)$) node { $y^{i+2}$};

    \node at (6.8, 0.25) { $f_1(x)$};
    \node at (0.02, 6.3) { $f_2(x)$};
    
\end{tikzpicture}
    \end{subfigure}
    \qquad \qquad
    \begin{subfigure}[t]{0.43\textwidth}
    \begin{tikzpicture}

    \draw (1,6) coordinate (Y1)  
        -- (1.7,4.5) coordinate (Y2)  
        -- (3,3) coordinate (C)  
        -- (4.7,1.5)  coordinate (Y3) 
        -- (6.2,1)  coordinate (Y4);
    
    \coordinate (Int) at ($(Y2) + 1.58*(0.7,-1.5)$);
    
    \pic [draw, fill = gray, opacity = 0.5] {angle = Y3--Int--Y2};
    \draw ($(Int) + (0.73,0.1)$) node {$\theta^{ij}$};
    \draw (Int) node { \textbullet};

    \draw[dashed] (Y2) -- ($(Y2) + 2*(0.7,-1.5)$);
    \draw[dashed] (Y3) -- ($(Y3) + 2*(-1.5,0.5)$);

    \draw[thick,->] (0,0.5) -- (7,0.5);
    \draw[thick,->] (0.5,0) -- (0.5,6.5);
    \draw ($(Y1) - (0,0.05)$) node { \textbullet};
    \draw ($(Y2) - (0,0.05)$) node { \textbullet};
    \draw ($(Y3) - (0,0.05)$) node { \textbullet};
    \draw ($(Y4) - (0,0.05)$) node { \textbullet};
    \draw ($(C) - (0,0.05)$) node { \textbullet};

    \draw ($(Y1) + (-0.1, -0.3)$) node { $y^{i}$};
    \draw ($(Y2) + (-0.2, -0.3)$) node { $y^{i+1}$};
    \draw ($(Y3) + (-0.1, -0.3)$) node { $y^{j}$};
    \draw ($(Y4) + (0.2, -0.2)$) node { $y^{j+1}$};

    \node at (6.8, 0.25) { $f_1(x)$};
    \node at (0.02, 6.3) { $f_2(x)$};
    
\end{tikzpicture}    
    \end{subfigure}
    \caption{Example of the angle measure in groups of size 2 (left) and in the general case (right).}
    \label{fig:angleMeasure}
\end{figure}

An almost straight group angle indicates that the supported points are nearly collinear (see Proposition \ref{prop:1}). 
Therefore, a grouping strategy should prioritize groups with the largest possible
group angles. 

\subsubsection{Group ND bound}
\label{sec:ndbound}

Due to the integrality of the objective function values, the maximum number of (unsupported or non-extreme supported) nondominated points within 
triangle $\tri{i}$ can be easily determined. 
This bound, referred to as the \emph{ND bound} measure of a group, is defined as follows:

\begin{equation}    
\beta^i = \min\left\{y^{i+1}_1 - y^i_1, y^i_2 - y^{i+1}_2\right\} - 1
\label{eq:triangleBound}
\end{equation}

Note that when $\beta^i=0$, triangle $\tri{i}$ is necessarily empty. 

The value of this measure for a triangle $\tri{i}$ can be refined if some of the nondominated points within it are known. 
Let $\left\{y^i, y^{i_1}, y^{i_2}, \ldots, y^{i_r}, y^{i+1}\right\} \subseteq Y_N \cap \tri{i}$, such that 
$y_1^i < y_1^{i_1} < y_1^{i_2} < \cdots < y_1^{i_r} < y_1^{i+1}$.
We redefine the ND bound $\beta^i$ as follows
$$
\beta^i =  
\min\left\{y^{i_1}_1 - y^i_1, y^{i_1}_2 - y^i_2\right\}+
\min\left\{y^{i+1}_1 - y^{i_r}_1, y^{i+1}_2 - y^{i_r}_2\right\}+
\sum_{k=2}^r 
    \min\left\{y^{i_{k}}_1 - y^{i_{k-1}}_1, y^{i_k}_2 - y^{i_{k-1}}_2\right\} - (r + 1)$$

Finally, this measure can be extended to a group of triangles $\gr{i}{j}$:
$$
\beta^{ij} = \sum_{k: \tri{k} \in \gr{i}{j}} \beta^k$$

By definition, this measure indicates the potential of a group to contain nondominated points.

\subsection{\textit{A priori} strategy}
\label{sub:apriori}

We refer to \emph{a priori} strategies as those that construct the entire grouping before exploring its groups.
The following subsections outline two main a priori strategies implemented in this study. 
The first strategy, \emph{merge-based grouping}, iteratively combines triangles to form groups, whereas the second strategy, \emph{splitting-based grouping}, begins with all triangles in a single group and iteratively divides it until a specified condition is met.

\subsubsection{Merge-based grouping}

Within the merge-based grouping strategy, we consider
two variants: \emph{fixed size} and \emph{greedy}.
We consider $m$ triangles, $\tri{1}, \hdots, \tri{m}$, that
are defined by $m+1$ supported points found in the first
phase.

\paragraph{Fixed-size variant}

Given a parameter $s$ that defines the size of each group, this variant forms a grouping by partitioning the set of $m$ triangles into $\lfloor \frac{m}{s} \rfloor$ groups, each of which with $s$ adjacent triangles.
If $\frac{m}{s}$ is not integer, one of the groups is allowed to contain $s^\prime = m-s \cdot \lfloor \frac{m}{s} \rfloor$ triangles. In our experiments, we consider the latter group to contain the $s^\prime$ triangles defined by the 
last $s^\prime+1$ supported points with the largest $f_1$ values.

\paragraph{Greedy variant}
\label{sec:GreedyVar}
Given a parameter $t \geq 2$, which defines the maximum allowable group size, the greedy-based variant uses a specific group 
measure to iteratively merge at most $t$ triangles greedily. The variant begins by merging $t$ consecutive triangles that optimize the given group measure. The merging process is continued until no such set of consecutive $t$ triangles can be found. 
In that case, $t$ is decremented, and the greedy selection process is repeated. 
The merging continues until $t$ reaches 1. From this point, the remaining triangles can only be isolated, forming groups of size 1 and terminating the process.
The group measures to be used as the greedy criterion can be the largest group angle value or the largest ND bound value.

\subsubsection{Splitting-based grouping}
\label{sec:split}

The splitting-based approach involves sequentially selecting supported points that serve as dividers between two adjacent groups. 
For the sake of explanation, we assume that all triangles $\tri{1}, \tri{2},\ldots,\tri{m}$ form initially a single group, $\gri{}{} = \gr{1}{m}$, and that this group will be split iteratively by following the next steps:

\begin{enumerate}
    \item Obtain the initial list of candidate split points $\zeta = \{y^2, y^3, \hdots, y^{m}\}$. 
    \item Let group $\gr{i}{j}$ be a group in $\gri{}{}$ with the largest number of triangles.
    \item Select, based on a given measure, which supported point $y^{c}$ will serve as the next splitting
    point, $i < c \leq j$. This point is removed from $\zeta$.
    \item Replace $\gr{i}{j}$ with $\gr{i,}{c-1}$ and $\gr{c}{j}$ in $\gr{}{}$. This is equivalent to splitting $\gr{i}{j}$ at point $y^c$.
    \item If the stopping criterion is not met, go to Step 2. Otherwise, return $\gri$.
\end{enumerate}

In Step 3, the group angle serves as an effective criterion for splitting, where the split point $y^c$ is chosen based on the smallest group angle. Due to the connection of this measure with collinearity (see Section 6.1.1), this approach ensures that the supported points within a group remain as close as possible to the collinear case. The approach terminates when every group in $\gri$ contains at most a given number of triangles or when an average group size is achieved.

\subsection{Dynamic strategy}
\label{sub:dynamic}

A dynamic strategy iteratively forms a new group using information gathered from previously explored groups.
A key ingredient in such a strategy is to use what we call \emph{coverage}, introduced by Steiner and Radzik~\cite{steiner2008computing} as a heuristic improvement. This concept, used when iteratively exploring triangles, was shown to be quite beneficial in the experiments reported by the authors. We propose an extension of this concept in two directions. The first, 
which is straightforward, applies coverage to groups rather than
to 
triangles only. The second 
uses the covering information not only for the current exploration of the group, but throughout the whole process, leading to the coverage of more triangles and a reduced exploration cost overall.

As in the original work~\cite{steiner2008computing}, a triangle $\tri{i}$ is covered if the ranking algorithm surpasses its upper bound, making further exploration of this triangle unnecessary. Otherwise, if the upper bound is not surpassed, $\tri{i}$ is considered \emph{partially covered} and must be explored in subsequent steps.
An example is shown in Figure \ref{fig:coverage}. Consider that the complete exploration of triangle $\tri{2}$ found all nondominated points with weighted-sum values lower than $w(\bar y)$. Therefore, all nondominated points within triangles $\tri{1}$ and $\tri{3}$ are also
found, and they can be ignored in further explorations. The remaining two triangles were not fully covered: $\tri{4}$ is partially covered, while $\tri{5}$ lies outside the search range. Both triangles 
would need to be further explored by the ranking algorithm. 

\begin{figure}[t]
\centering
\begin{tikzpicture}
    \coordinate (s1) at (1.3, 6.4);
    \coordinate (s2) at (1.5, 5.5);            
    \coordinate (s3) at (2.2, 3.5);
    \coordinate (s4) at (3, 2.5);
    \coordinate (s5) at (3.9, 1.6);
    \coordinate (s6) at (5.3, 0.7);
    \coordinate (disp) at (0,-0.02);

    \fill[gray!40] (s1) -- (s1 -| s2) -- (s2);
    \fill[gray!40] (s3) -- (s3 -| s4) -- (s4);
    \fill[gray!10] (s4) -- ($(s4) + (0.5,0)$) -- ($(s5) + (0,0.1)$) -- (s5);

    \fill[gray!90] (s2) -- ($(s2) + (0.5,0)$) -- ($(s3) + (0,1.59)$) -- (s3) -- cycle;
    
    \node at ($(s1) + (disp)$) {\textbullet};
    \node at ($(s2) + (disp)$)  {\textbullet};
    \node at ($(s3) + (disp)$)  {\textbullet};
    \node at ($(s4) + (disp)$)  {\textbullet};
    \node at ($(s5) + (disp)$)  {\textbullet};
    \node at ($(s6) + (disp)$)  {\textbullet};

    \node at ($(s1) - (0.2,0.5)$) {$\tri{1}$};
    \node at ($(s2) - (-0.05,0.9)$) {$\tri{2}$};
    \node at ($(s3) - (-0.15,0.65)$) {$\tri{3}$};
    \node at ($(s4) - (-0.25,0.6)$) {$\tri{4}$};
    \node at ($(s5) - (-0.7,0.75)$) {$\tri{5}$};

    \draw (s1) -- (s2) -- (s3) -- (s4) -- (s5) -- (s6);
    
    \draw (s1) -- (s1 -| s2) -- (s2) -- (s2 -| s3) -- (s3) -- (s3 -| s4) -- (s4) -- (s4 -| s5) -- (s5) -- (s5 -| s6) -- (s6);

    \draw[dashed] ($(s2 -| s3) - 0.5*(1.5,-3) - (-0.05, 0.5)$) -- ($(s2 -| s3) - (-0.05, 0.5) + 1.5*(1.5,-3)$);

    \node at ($(s2 -| s3) - (0.12, 0.2)$) { \textbullet};
    \node at ($(s2 -| s3) + (0.12, 0.17)$) {$\bar y$};    \node at ($(s2 -| s3) + (0.82, -1.17)$) {$w(\bar y)$};

    \draw[thick,->] (0,0.5) -- (7,0.5);
    \draw[thick,->] (0.5,0) -- (0.5,6.5);
                \node at (6.8, 0.25) { $f_1(x)$};
            \node at (0, 6.3) {$f_2(x)$};

\end{tikzpicture}
\caption{Example of covered, partially covered, and uncovered triangles}
\label{fig:coverage}
\end{figure}
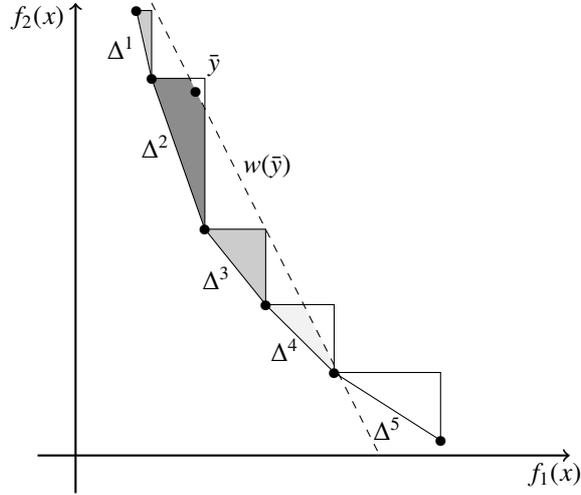

More formally, to extend this concept to groups, let us consider a triangle $\tri{i}$ that is not part of the group $\gr{j}{k}$ being explored. When finding, during this exploration, nondominated points belonging to $\tri{i}$, 
set $U^i$ of its local upper bounds 
can be updated (see relation~\eqref{eq:U}). Let  $v^{jk}$ be the value
at which the ranking algorithm stopped when finishing exploring group $\gr{j}{k}$. If $w^{jk}(u) < v^{jk}$ for all $u \in U^i$, then all the nondominated points belonging to $\tri{i}$ have been found during the exploration. 

Steiner and Radzik empirically show that this additional coverage step is useful in certain cases. 
However, they also note that if the condition is not met for all local upper bounds of $U^i$, the additional effort becomes useless, as triangle $\tri{i}$ must still be explored from scratch in a further iteration. We claim that the exploration of
multiple groups nearby $\tri{i}$, without discarding the nondominated points found within, may allow this triangle to be covered.

For this purpose, let $\bar{U^i} \subset U^i$ be the set of \emph{active} local upper bounds such that $u \in \bar{U^i}$ if $w^{jk}(u) \geq v^{jk}$. The exploration of $\gr{j}{k}$ guarantees that all nondominated points in $\tri{i}$ that are upper bounded by any $u \in U^i \setminus \bar{U^i}$ have been found. Therefore, nondominated points in $\tri{i}$ can only lie in the zones upper bounded by $u \in \bar{U^i}$. 
It is possible that, when exploring another group, say $\gr{\ell}{n}$, the coverage of $\tri{i}$ can be improved or, ideally, fully covered. 
This is achieved if $w^{\ell n}(u) < v^{\ell n}$ for some or all $u \in \bar{U^i}$, where $v^{\ell n}$ is the value at which the ranking algorithm is stopped when finishing exploring $\gr{\ell}{n}$. Observe that $\bar{U^i}$ is refined each time a new nondominated point is found in $\tri{i}$ while exploring $\gr{\ell}{n}$. We refer to this approach as \emph{extended coverage}.

\section{Experimental Analysis}
\label{sec:results}

In this section, we describe an experimental analysis of our grouping strategies on the \acrfull*{bomst} problem. 
Given a graph $G = (V,E)$, and costs $(c_1(e),c_2(e)) \in \mathbb Z^2_>$ for each edge $e \in E$, the goal in the~\acrshort*{bomst} problem is to find the set $Y_N$ of all nondominated points, each corresponding to an efficient spanning tree.
While $Y_N$ is known to have exponential size~\cite{hamacher1994spanning}, the set $Y_{NSE}$ is polynomially bounded~\cite{Seipp2013}.
An in-depth experimental analysis of current solution approaches to the~\acrshort*{bomst} problem can be found in~\cite{fernandes2020empirical}. 
The experimental results obtained by the authors indicate that the two-phase method incorporating a ranking algorithm in the second phase, as proposed in \cite{steiner2008computing}, is among the best-performing approaches for most of the instances.

\subsection{Experimental setup}
\label{sec:expSetup}

In our experimental analysis, we consider complete graphs. 
The generation of each edge cost follows the procedure described in~\cite{verel2013structure}, which allows for varying the degree of conflict between the two objectives by a certain correlation factor defined a priori. 
Each type of instance is then characterized by the three following parameters:
\begin{itemize}
    \item Size ($n$): The number of nodes in each graph. 
    Values for $n$ are set to $50$, $100$, and $150$.
    \item Range ($r$): The range of the cost of the edges.
    We only vary the upper limit, and we fix the lower 
    limit to 1.    
    The considered values of $r$ are $10^2$, $10^3$ and $10^4$.
    \item Correlation factor ($\rho$): the values of the correlation matrix as in~\cite{verel2013structure}, used to generate costs on the edges. 
    The chosen values for $\rho$ are $0.8$, $0$, and $-0.8$ reflecting increasing degrees of conflict between the two objectives.
\end{itemize}

There are $27$ possible combinations for the specified values on each parameter, with 10 instances generated for each combination.
Different sizes and correlation values are oftentimes considered in the literature, but the range is fixed \cite{fernandes2020empirical, sourd2008multiobjective,steiner2008computing}. 
However, the maximum value for edge costs seems to have a significant impact on the performance 
of different grouping strategies, as reported in our experimental
analysis.

The first phase was performed with a dichotomic approach which solves $2\cdot |Y_{NSE}|-1$ instances of the single objective minimum spanning tree problem. 
Moreover, this phase is common to all strategies and is not considered in our performance comparisons.
We used the ranking algorithm proposed by Gabow \cite{gabow1977two}, which takes $O(k \cdot |E| \cdot \alpha(|E|,|V|)+|E| \log |E|)$ time to enumerate the $k$ best spanning trees where $\alpha$ is the inverse Ackermann function.
We consider that the number of visited solutions is an appropriate measure of the computational effort. This measure is more robust than the CPU time, since it is invariant over different machine setups and not subject to external interference. Moreover, it is strongly correlated with the running time since most of the computation time is spent on the enumeration of solutions and, for a given instance, the time to enumerate a new solution in the ranking is approximately constant.

The experiments were run in a cluster with 2 Intel Xeon Silver 4210R 2.4 GHz, 20 total cores, 2 threads each, 256 GB RAM, with operating system Debian GNU/Linux 10 6.1.0-32-amd64. 
The implementations were coded in {\tt{C++}} and compiled with {\tt{g++}} 9.4.0 with {\tt{-O3}} compilation flag. Instances are publicly available at \url{https://github.com/FlipM/BOMST_Benchmark}.

\subsection{Analysis of optimal groupings}
\label{sec:optRes}

In this section, we provide an analysis of the optimal groupings for the \acrshort*{bomst} instances using the grouping graph introduced in Section~\ref{sec:grouping}. 
To this end, we value each arc of the grouping graph with the cost of exploring the corresponding group, which, in our case, is the number of solutions visited by the ranking algorithm in order to guarantee that all nondominated points belonging to this group have been found. 

Computing these values for all arcs is very costly. 
We show that it is unnecessary to compute all these values to obtain the optimal grouping. For this purpose, we first create all arcs of type $(i,i+1)$, $i=1,..,m$, which correspond to all groups of size 1, and value them by applying the ranking algorithm. We then consider the generation and evaluation of all arcs of type $(i,i+2)$, $i=1,..,m-1$. For each such arc $(i,j)$, we first compute the shortest path value $u_{ij}$ using the already known and valued arcs. Then, when applying the ranking algorithm to the corresponding group $\gr{i,}{j-1}$, we create arc $(i,j)$, valued by $v_{ij}$ if it enumerates $v_{ij} < u_{ij}$ solutions. Otherwise,  the enumeration stops as soon as $v_{ij} = u_{ij}$, and arc $(i,j)$ is not created. We continue similarly for arcs $(i,i+3),\ldots,(i,i+m)$.

A second simplification is to set a limit on the size of each group. Preliminary experiments reveal that, while grouping a few triangles is often quite beneficial, grouping many triangles leads to prohibitive computational efforts. It is then possible to define a maximum number $\tau$ of triangles that could belong to a group (our experiments show that the largest groups very rarely contain more than 7 triangles). This has the double advantage of defining a reduced grouping graph $G_{\tau} = (V,A_{\tau})$ with $A_{\tau} = \{(i,j) \in V \times V: i<j \mbox{ where } j-i \leq \tau \}$ making unnecessary the evaluation of the measure on arcs from $A \setminus A_{\tau}$ and improving the determination of an optimal path which is performed now in $\mathcal{O}(m)$ time.

\begin{sidewaystable}
\centering
\setlength{\tabcolsep}{4pt}
\begin{tabular}{rrrr|rrrrr|rrrrrrrr}
\multicolumn{4}{c|}{Instances}                                 & \multicolumn{5}{c|}{Results}                           & \multicolumn{8}{c}{Group size}                        \\
$r$                     & $\rho$                & $n$ & \#s & $|Y_{NSE}|$ & $|Y_N|^B$ & $|Y_N|$  & Optimal   & F1    & 1    & 2    & 3    & 4    & 5   & 6   & 7   & Average \\ \hline
\multirow{9}{*}{$10^2$} & \multirow{3}{*}{0.8}  & 50  & 10     & 28.0        & 94.6      & 86.8     & 328.6     & 1.287 & 6.2  & 4.2  & 2.2  & 0.2  & 0.0 & 0.0 & 0.0 & 1.72    \\
                        &                       & 100 & 10     & 34.7        & 118.4     & 118.0    & 2 229.8   & 1.062 & 18.6 & 2.6  & 0.7  & 0.0  & 0.1 & 0.1 & 0.0 & 1.22    \\
                        &                       & 150 & 9      & 39.8        & 145.0     & 139.9    & 168 734.8 & 1.059 & 24.4 & 1.4  & 0.4  & 0.1  & 0.0 & 0.1 & 0.0 & 1.11    \\ \cline{2-17} 
                        & \multirow{3}{*}{0.0}  & 50  & 10     & 100.3       & 762.8     & 648.9    & 2 815.4   & 1.639 & 7.8  & 16.8 & 12.4 & 3.3  & 0.4 & 0.0 & 0.0 & 2.30    \\
                        &                       & 100 & 10     & 185.4       & 1 160.3   & 1 118.0  & 6 954.5   & 1.450 & 28.6 & 34.1 & 19.6 & 3.1  & 0.6 & 0.4 & 0.1 & 2.01    \\
                        &                       & 150 & 10     & 250.7       & 1 499.7   & 1 491.4  & 27 398.6  & 1.301 & 78.0 & 35.3 & 22.2 & 2.5  & 0.3 & 0.7 & 0.9 & 1.70    \\ \cline{2-17} 
                        & \multirow{3}{*}{-0.8} & 50  & 10     & 154.5       & 2 690.8   & 2 246.9  & 19 491.1  & 1.606 & 14.6 & 27.8 & 17.8 & 5.7  & 0.5 & 0.0 & 0.0 & 2.24    \\
                        &                       & 100 & 9      & 309.5       & 5 037.7   & 4 824.3  & 235 384.3 & 1.391 & 83.1 & 45.6 & 34.4 & 5.4  & 0.7 & 0.0 & 0.0 & 1.79    \\
                        &                       & 150 & 0      & 425.2       & 7 298.6   & -        & -         & -     & -    & -    & -    & -    & -   & -   & -   & -       \\ \hline
\multirow{9}{*}{$10^3$} & \multirow{3}{*}{0.8}  & 50  & 10     & 37.9        & 934.0     & 247.5    & 1 255.7   & 1.673 & 1.9  & 4.8  & 4.4  & 2.0  & 0.6 & 0.1 & 0.0 & 2.63    \\
                        &                       & 100 & 10     & 87.5        & 1 320.2   & 781.8    & 4 518.5   & 1.647 & 3.6  & 11.1 & 11.7 & 4.8  & 0.7 & 0.1 & 0.0 & 2.63    \\
                        &                       & 150 & 10     & 132.1       & 1 507.8   & 1 218.8  & 7 646.6   & 1.597 & 7.6  & 18.9 & 17.4 & 6.0  & 1.1 & 0.2 & 0.0 & 2.51    \\ \cline{2-17} 
                        & \multirow{3}{*}{0.0}  & 50  & 10     & 113.9       & 7 345.3   & 1 469.5  & 8 662.2   & 1.972 & 2.9  & 13.5 & 15.2 & 7.0  & 1.5 & 0.2 & 0.0 & 2.78    \\
                        &                       & 100 & 10     & 267.1       & 11 840.3  & 5 040.2  & 32 257.5  & 1.981 & 5.4  & 29.7 & 40.4 & 15.5 & 3.0 & 0.1 & 0.0 & 2.80    \\
                        &                       & 150 & 10     & 444.0       & 15 017.6  & 9 217.1  & 59 150.1  & 2.051 & 10.3 & 47.6 & 59.5 & 28.4 & 6.9 & 0.8 & 0.0 & 2.85    \\ \cline{2-17} 
                        & \multirow{3}{*}{-0.8} & 50  & 10     & 184.0       & 26 716.1  & 6 246.5  & 57 694.4  & 1.901 & 6.2  & 20.6 & 24.7 & 11.6 & 2.3 & 0.3 & 0.2 & 2.77    \\
                        &                       & 100 & 10     & 445.0       & 50 604.4  & 21 836.4 & 239 891.8 & 1.921 & 12.4 & 47.8 & 61.4 & 28.4 & 6.4 & 0.5 & 0.2 & 2.81    \\
                        &                       & 150 & 0      & 712.7       & 73 383.8  & -      & -       & -     & -    & -    & -    & -    & -   & -   & -   & -       \\ \hline
\multirow{9}{*}{$10^4$} & \multirow{3}{*}{0.8}  & 50  & 10     & 37.8        & 9 300.3   & 256.6    & 1 425.1   & 1.749 & 1.2  & 5.7  & 4.1  & 2.2  & 0.5 & 0.1 & 0.0 & 2.67    \\
                        &                       & 100 & 10     & 95.5        & 13 429.3  & 1 443.5  & 10 520.8  & 1.841 & 2.3  & 11.7 & 10.8 & 6.3  & 1.9 & 0.2 & 0.0 & 2.83    \\
                        &                       & 150 & 10     & 150.7       & 14 728.5  & 3 036.2  & 19 521.4  & 1.835 & 3.0  & 15.7 & 19.0 & 9.6  & 3.2 & 0.3 & 0.2 & 2.92    \\ \cline{2-17} 
                        & \multirow{3}{*}{0.0}  & 50  & 10     & 113.6       & 83 511.4  & 1 773.5  & 11 398.0  & 2.025 & 2.4  & 11.8 & 15.8 & 7.4  & 1.4 & 0.2 & 0.2 & 2.87    \\
                        &                       & 100 & 10     & 276.7       & 119 201.1 & 7 486.4  & 57 365.2  & 2.073 & 3.6  & 29.5 & 36.3 & 19.9 & 4.3 & 0.4 & 0.0 & 2.93    \\
                        &                       & 150 & 10     & 453.6       & 148 869.5 & 15 893.8 & 134 437.2 & 2.101 & 6.2  & 37.6 & 65.9 & 34.6 & 6.9 & 0.0 & 0.0 & 2.99    \\ \cline{2-17} 
                        & \multirow{3}{*}{-0.8} & 50  & 10     & 183.9       & 269 951.9 & 7 845.0  & 77 565.7  & 1.929 & 4.4  & 21.3 & 25.0 & 12.5 & 1.9 & 0.2 & 0.0 & 2.80    \\
                        &                       & 100 & 10     & 451.4       & 506 843.1 & 32 116.6 & 334 789.5 & 2.009 & 7.4  & 41.6 & 60.0 & 33.4 & 8.9 & 0.2 & 0.0 & 2.97    \\
                        &                       & 150 & 0      & 743.7       & 732 904.0 & -        & -         & -     & -    & -    & -    & -    & -   & -   & -   & -      
\end{tabular}
    \caption{Results for optimal grouping.}
    \label{table:newOptF1}
\end{sidewaystable}

Table~\ref{table:newOptF1} shows the results for the $270$ instances, grouped by size, correlation, and range. The \textit{Instances} column group refers to the instance parameters mentioned in Section~\ref{sec:expSetup} plus column \emph{\#s}, which shows the number of instances in which an optimal grouping was found within a $5$-hour time limit. In particular, no optimal grouping was found for any of the instances with $n=150$ and $\rho = -0.8$. In those cases, we only report information from the first phase and ignore those instances in the subsequent experiments.

The \emph{Results} column group reports, for each type of instance, the number of nondominated points (column \emph{$|Y_N|$}) and supported extreme points (column {$|Y_{NSE}|$}), averaged over the solved instances.
Column \emph{$|Y_N|^B$} represents the average upper bound on the number of nondominated points,  
considering the supported extreme points from the first phase. It is calculated as $|Y_N|^B = |Y_{NSE}| + \sum_{i = 1}^m \beta^i$, where $\beta^i$ is defined in Eq.~\eqref{eq:triangleBound}.
Column \emph{Optimal} shows the average number of enumerated solutions by the optimal grouping.

To evaluate the performance of each strategy, we first define the \emph{effectiveness ratio}  for each instance, calculated as the ratio
of the number of enumerated solutions using the grouping strategy to the number obtained with the optimal grouping.
Each row of column \emph{F1} shows the harmonic mean of the effectiveness ratios using the baseline strategy, which explores each triangle individually. Note that the harmonic mean is calculated over the solved instances for each corresponding combination of parameters.
Information on the group sizes of optimal groupings is also detailed in Table \ref{table:newOptF1}. 
Each column, labeled from \emph{1} to \emph{7}, provides the average number of groups of the corresponding size across the solved instances for each type of instance.
Column \emph{Average} shows the average group size for the optimal groupings.

Table~\ref{table:newOptF1} indicates that the $F1$ approach generally explores nearly twice as many solutions as the optimal grouping, except when the range is small, where the ratio is significantly lower. Additionally, the optimal groupings tend to favor groups of size 2 and 3, except for instances with small ranges, where exploring single triangles appears more effective. Moreover, the performance of the $F1$ approach appears largely unaffected by the instance size or the correlation between objectives.

These experimental results suggest that the range parameter significantly influences the sizes of the groups in the optimal groupings. When the range is small, there are few groups with a size of 2 or more, suggesting that grouping may have limited effectiveness. Table \ref{table:newOptF1} shows that the bound $|Y_N|^B$ is very tight for instances with small ranges ($|Y_N|$ is at most $20\%$ lower than the bound), suggesting that the number of nondominated points within each triangle $\tri{i}$ is very close to $\beta^i$. As a result, these nondominated points are likely to be concentrated near the boundary $bd(\mathcal C_N)$. Consequently, the likelihood of revisiting solutions during the exploration of different triangles is lower. For such cases, grouping should be unnecessary.

Conversely, grouping can consistently reduce the computational effort by approximately half for larger ranges. In these cases, the bound $|Y_N|^B$ is less tight, indicating that nondominated points are likely to be farther away from $bd(\mathcal C_N)$. This increases the likelihood of revisiting solutions when exploring different triangles, justifying the need for grouping to minimize redundancy.

This analysis is also reflected in the group size columns. Optimal groupings in small-range instances rarely require groups of size 2 or more. In instances with larger ranges, isolated triangles are less used than groups of sizes 2 and 3. Groups of size 5 or more are seldom used in optimal groupings, and groups of size 7 are 
nearly never used.
Thus, the optimal average group size oscillates between 1 and 3, with a large majority of cases being above 2. These values will be useful to build the strategies presented in the next section.

\subsection{Analysis of a priori strategies}
\label{sec:AprioriRes}

This section reports information about the a priori strategies that performed better in our study. 
We divide the proposed strategies into categories based on variants presented in Section \ref{sub:apriori}:

\begin{itemize}
    \item F1-F4: Fixed grouping strategy with the group size corresponding to the number after `F'. Note that F1 corresponds to the same F1
    in Table \ref{table:newOptF1}, that is,
    no grouping is considered.
    \item SA2.0 and SA2.5:  Splitting strategies that choose the supported point with the smallest group angle as the split point (see Section~\ref{sec:split}). The suffixes $2.0$ and $2.5$ represent 
    the average group size used as the stopping criterion. The values are chosen based on the \emph{Average} column in Table \ref{table:newOptF1}.
    \item GA2/3 and GN2/3: Greedy variant as defined in Section \ref{sec:GreedyVar},
    using group angle (prefix GA) and ND bound (prefix GN). In each iteration, adjacent triangles that maximize the selected measure are merged into a group of size $t$. The suffix (2 or 3) denotes the value of the parameter $t$.
\end{itemize}                       

The performance of each strategy over the instance benchmark is shown in Table~\ref{tab:newTabAprior}. Column \emph{Optimal} presents the average number of solutions enumerated by the ranking algorithm considering the optimal grouping (as shown in Table~\ref{table:newOptF1}).  
Column \textit{mean} shows the harmonic mean of the effectiveness ratio, while \textit{\#s}
indicates the number of solved instances for the corresponding problem type.
For each type of instance, the lowest mean of the effectiveness ratio is
highlighted in bold.

\begin{sidewaystable}

    \centering
\setlength{\tabcolsep}{4pt}
\begin{tabular}{rrrr|rr rr rr rr|rr rr|rr rr rr rr}
\multicolumn{3}{r}{Instances}                                  & \multirow{2}{*}{Optimal} & \multicolumn{2}{c}{F1} & \multicolumn{2}{c}{F2} & \multicolumn{2}{c}{F3} & \multicolumn{2}{c|}{F4} & \multicolumn{2}{c}{SA2.0} & \multicolumn{2}{c|}{SA2.5} & \multicolumn{2}{c}{GA2} & \multicolumn{2}{c}{GA3} & \multicolumn{2}{c}{GN2} & \multicolumn{2}{c}{GN3} \\
r                       & $\rho$                & n   &                          & mean             & \#s  & mean             & \#s  & mean             & \#s  & mean          & \#s      & mean              & \#s    & mean               & \#s    & mean             & \#s   & mean             & \#s   & mean          & \#s      & mean         & \#s       \\ \hline
\multirow{8}{*}{$10^2$} & \multirow{3}{*}{0.8}  & 50  & 328.6                    & 1.287            & 10  & 1.353            & 10  & 2.357            & 10  & 5.533         & 10      & 1.475             & 10    & 2.063              & 10    & \textbf{1.150}   & 10   & 1.697            & 10   & 1.493         & 10      & 2.322        & 10       \\
                        &                       & 100 & 2 229.8                  & \textbf{1.062}   & 10  & 3.944            & 10  & 17.343           & 8   & 157.294       & 6       & 6.005             & 10    & 20.982             & 9     & 3.012            & 10   & 12.011           & 9    & 5.872         & 10      & 58.252       & 8        \\
                        &                       & 150 & 168 734.8                & \textbf{1.059}   & 9   & 16.005           & 1   & -              & 0   & -           & 0       & 7.036             & 4     & -                & 0     & 7.797            & 3    & 21.874           & 1    & 125.603       & 1       & -            & 0        \\ \cline{2-24} 
                        & \multirow{3}{*}{0.0}  & 50  & 2 815.4                  & 1.639            & 10  & \textbf{1.248}   & 10  & 1.530            & 10  & 2.547         & 10      & 1.294             & 10    & 1.429              & 10    & 1.274            & 10   & 1.328            & 10   & 1.304         & 10      & 1.673        & 10       \\
                        &                       & 100 & 6 954.5                  & 1.450            & 10  & 1.403            & 10  & 2.465            & 10  & 5.765         & 10      & 1.464             & 10    & 2.580              & 10    & \textbf{1.335}   & 10   & 1.924            & 10   & 1.477         & 10      & 2.658        & 10       \\
                        &                       & 150 & 27 398.6                 & \textbf{1.301}   & 10  & 2.445            & 10  & 9.164            & 9   & -           & 0       & 1.679             & 10    & 3.410              & 10    & 1.899            & 10   & 5.594            & 9    & 2.910         & 10      & 15.848       & 6        \\ \cline{2-24} 
                        & \multirow{2}{*}{-0.8} & 50  & 19 491.1                 & 1.606            & 10  & 1.306            & 10  & 1.731            & 10  & 3.580         & 10      & 1.369             & 10    & 1.699              & 10    & \textbf{1.274}   & 10   & 1.652            & 10   & 1.389         & 10      & 2.247        & 10       \\
                        &                       & 100 & 235 384.3                & \textbf{1.391}   & 9   & 2.011            & 8   & 7.062            & 3   & 12.886        & 1       & 2.144             & 9     & 3.758              & 5     & 1.772            & 9    & 2.804            & 7    & 2.277         & 8       & 6.863        & 3        \\ \hline
\multirow{8}{*}{$10^3$} & \multirow{3}{*}{0.8}  & 50  & 1 255.7                  & 1.673            & 10  & 1.259            & 10  & 1.372            & 10  & 2.131         & 10      & 1.230             & 10    & 1.248              & 10    & 1.273            & 10   & \textbf{1.214}   & 10   & 1.326         & 10      & 1.436        & 10       \\
                        &                       & 100 & 4 518.5                  & 1.647            & 10  & \textbf{1.245}   & 10  & 1.299            & 10  & 1.941         & 10      & 1.250             & 10    & 1.373              & 10    & 1.254            & 10   & 1.281            & 10   & 1.305         & 10      & 1.569        & 10       \\
                        &                       & 150 & 7 646.6                  & 1.597            & 10  & 1.240            & 10  & 1.463            & 10  & 2.351         & 10      & \textbf{1.228}    & 10    & 1.334              & 10    & 1.263            & 10   & 1.301            & 10   & 1.292         & 10      & 1.717        & 10       \\ \cline{2-24} 
                        & \multirow{3}{*}{0.0}  & 50  & 8 662.2                  & 1.972            & 10  & 1.290            & 10  & \textbf{1.243}   & 10  & 1.820         & 10      & 1.351             & 10    & 1.364              & 10    & 1.382            & 10   & 1.288            & 10   & 1.380         & 10      & 1.456        & 10       \\
                        &                       & 100 & 32 257.5                 & 1.981            & 10  & 1.315            & 10  & \textbf{1.282}   & 10  & 1.644         & 10      & 1.357             & 10    & 1.361              & 10    & 1.393            & 10   & 1.289            & 10   & 1.396         & 10      & 1.441        & 10       \\
                        &                       & 150 & 59 150.1                 & 2.051            & 10  & 1.363            & 10  & \textbf{1.306}   & 10  & 1.675         & 10      & 1.418             & 10    & 1.421              & 10    & 1.466            & 10   & 1.358            & 10   & 1.436         & 10      & 1.467        & 10       \\ \cline{2-24} 
                        & \multirow{2}{*}{-0.8} & 50  & 57 694.4                 & 1.901            & 10  & 1.327            & 10  & 1.336            & 10  & 1.830         & 10      & 1.318    & 10    & 1.322              & 10    & 1.362            & 10   & \textbf{1.305}            & 10   & 1.384         & 10      & 1.448        & 10       \\
                        &                       & 100 & 239 891.8                & 1.921            & 10  & \textbf{1.332}   & 10  & 1.353            & 10  & 1.802         & 10      & 1.359             & 10    & 1.390              & 10    & 1.390            & 10   & 1.343            & 10   & 1.408         & 10      & 1.517        & 10       \\ \hline
\multirow{8}{*}{$10^4$} & \multirow{3}{*}{0.8}  & 50  & 1 425.1                  & 1.749            & 10  & 1.238            & 10  & 1.323            & 10  & 1.917         & 10      & 1.213             & 10    & \textbf{1.195}     & 10    & 1.252            & 10   & 1.210            & 10   & 1.298         & 10      & 1.444        & 10       \\
                        &                       & 100 & 10 520.8                 & 1.841            & 10  & 1.302            & 10  & 1.335            & 10  & 1.812         & 10      & 1.295             & 10    & 1.291              & 10    & 1.335            & 10   & \textbf{1.245}   & 10   & 1.371         & 10      & 1.503        & 10       \\
                        &                       & 150 & 19 521.4                 & 1.835            & 10  & 1.276            & 10  & 1.302            & 10  & 1.871         & 10      & 1.252             & 10    & 1.244              & 10    & 1.303            & 10   & \textbf{1.211}   & 10   & 1.374         & 10      & 1.467        & 10       \\ \cline{2-24} 
                        & \multirow{3}{*}{0.0}  & 50  & 11 398.0                 & 2.025            & 10  & 1.339            & 10  & \textbf{1.256}   & 10  & 1.681         & 10      & 1.372             & 10    & 1.349              & 10    & 1.386            & 10   & 1.301            & 10   & 1.414         & 10      & 1.422        & 10       \\
                        &                       & 100 & 57 365.2                 & 2.073            & 10  & 1.320            & 10  & \textbf{1.239}   & 10  & 1.542         & 10      & 1.366             & 10    & 1.339              & 10    & 1.431            & 10   & 1.287            & 10   & 1.427         & 10      & 1.402        & 10       \\
                        &                       & 150 & 134 437.2                & 2.101            & 10  & 1.353            & 10  & \textbf{1.260}   & 10  & 1.500         & 10      & 1.382             & 10    & 1.366              & 10    & 1.467            & 10   & 1.284            & 10   & 1.453         & 10      & 1.377        & 10       \\ \cline{2-24} 
                        & \multirow{2}{*}{-0.8} & 50  & 77 565.7                 & 1.929            & 10  & 1.293            & 10  & 1.284            & 10  & 1.777         & 10      & 1.329             & 10    & 1.360              & 10    & 1.372            & 10   & \textbf{1.278}   & 10   & 1.384         & 10      & 1.446        & 10       \\
                        &                       & 100 & 334 789.5                & 2.009            & 10  & 1.326            & 10  & \textbf{1.273}   & 10  & 1.643         & 10      & 1.355             & 10    & 1.347              & 10    & 1.408            & 10   & 1.298            & 10   & 1.433         & 10      & 1.454        & 10      
\end{tabular}
    \caption{Results for the a priori strategies.}
    \label{tab:newTabAprior}
\end{sidewaystable}

We divide the analysis of the results in 
Table~\ref{tab:newTabAprior} in two parts:
the first for small-range instances ($r = 10^2$), and the second for instances with $r \geq 10^3$. Recall that, for the former, the nondominated points typically lie close to the boundary of $\mathcal C_N$, and are nearly supported (see Section
\ref{sec:optRes}). As a result, grouping
may bring limited efficiency.

For small-range instances, fixed methods with a size greater than one perform poorly, with 
results deteriorating as the group size increases. 
Although generally less effective than F1
in terms of effectiveness ratio, measure-based strategies solve more instances and exhibit a higher effectiveness ratio than F2, F3, and F4. Notably, GA2 outperforms F1 in three instance types and times out on only 8 out of the 80 instances, making it the second-best grouping strategy overall.
Among the strategies that allow groups of size $3$, only SA2.0 performed reasonably well, with $7$ timeouts.

For the remaining instances ($r \geq 10^3$), F2 and F3 are highly effective.
Compared to F1, F3 can reduce the computational effort by between $9\%$ and $40\%$, while F2 achieves reductions ranging from 22\% to 36\%. A decline in performance is noticeable with F4, which is worse than F1 in more than half of the instances and, when it does improve, it never exceeds $26\%$.
We tested fixed-size variants with groups of larger sizes, but due to their poor performance and frequent timeouts, the corresponding results are not included in this paper.

Both splitting strategies performed well, consistently improving upon F1 from $16\%$ to $35\%$. However, no clear advantage can be established between SA2.0 and SA2.5. Among the greedy strategies, GA3 achieved the best average performance in $5$ instance types, and the second-best in $9$ others. Its improvements over F1 range from $18\%$ to $38\%$. GA2 yielded a slightly lower performance. Both GA2 and GA3 are competitive with F2 and F3 and outperform the splitting variant. In contrast, both GN2 and GN3 variants performed poorly.

\subsection{Analysis of dynamic strategies}

We consider two dynamic variants that apply coverage to single triangles only:

\begin{itemize}
    \item SRKB4: The \emph{KB4} approach proposed by Steiner and Radzik~\cite{steiner2008computing}, 
    which was reported by the authors to be the best-performing two-phase method using ranking algorithms. 
    This approach uses simple coverage as introduced in Section~\ref{sub:dynamic}. 
    The authors suggest selecting the next triangle based on its \emph{size}. 
    In our study, we prioritize the triangle with the largest ND bound.
    \item ECU: The \emph{Extended Coverage}
    version of SRKB4, as detailed in Section~\ref{sub:dynamic}. 
    The ND bound of partially covered triangles is updated, and the exploration ordering of triangles is adjusted accordingly.
\end{itemize}

We also consider extended coverage applied to groups of triangles defined by two greedy variants (see Section \ref{sec:GreedyVar}):

\begin{itemize}
    \item GAEC2/3, which iteratively combines variant GA2/3 with extended coverage.
    \item GNECU2/3, which iteratively combines variant GN2/3 with 
    extended coverage and 
    update of ND bounds.
\end{itemize}

The results of the dynamic strategies are shown in Table \ref{tab:newTabDynamic}. Similarly to Table~\ref{tab:newTabAprior}, the number of solved instances and harmonic means of effectiveness ratios with respect to the optimal grouping are presented for each dynamic variant. For reference, we also report  the results for F1 and the best
results obtained with the best-performing 
a priori strategy for each instance type
(column \emph{Best AP}).

\begin{sidewaystable}
    \centering
\setlength{\tabcolsep}{4pt}
\begin{tabular}{rrrr|rrrr|rrrr|rrrrrrrr}
\multicolumn{3}{c}{Instance}                                                                & \multicolumn{1}{c|}{\multirow{2}{*}{Optimal}} & \multicolumn{2}{c}{F1} & \multicolumn{2}{c|}{Best AP} & \multicolumn{2}{c}{SRKB4} & \multicolumn{2}{c|}{ECU} & \multicolumn{2}{c}{GAEC2} & \multicolumn{2}{c}{GAEC3} & \multicolumn{2}{c}{GNECU2} & \multicolumn{2}{c}{GNECU3} \\
r                                      & \multicolumn{1}{l}{$\rho$} & \multicolumn{1}{l}{n} & \multicolumn{1}{c|}{}                         & mean        & \#s       & mean           & \#s          & mean              & \#s    & mean              & \#s   & mean              & \#s    & mean              & \#s    & mean               & \#s    & mean           & \#s        \\ \hline
\multirow{8}{*}{10\textasciicircum{}2} & \multirow{3}{*}{0.8}       & 50                    & 328.6                                         & 1.287       & 10       & 1.150          & 10          & \textbf{1.130}    & 10    & 1.132             & 10   & 1.213             & 10    & 1.660             & 10    & 1.449              & 10    & 2.312          & 10        \\
                                       &                            & 100                   & 2 229.8                                       & 1.062       & 10       & 1.062          & 10          & 1.019             & 10    & \textbf{1.015}    & 10   & 2.425             & 10    & 11.732            & 9     & 5.859              & 10    & 58.246         & 8         \\
                                       &                            & 150                   & 168 734.8                                     & 1.059       & 9        & 1.059          & 9           & \textbf{1.024}    & 9     & 1.025             & 9    & 7.790             & 3     & 21.874            & 1     & 125.595            & 1     & -              & 0         \\ \cline{2-20} 
                                       & \multirow{3}{*}{0.0}       & 50                    & 2 815.4                                       & 1.639       & 10       & 1.248          & 10          & 1.281             & 10    & 1.248             & 10   & \textbf{1.200}             & 10    & 1.303             & 10    & 1.220     & 10    & 1.612          & 10        \\
                                       &                            & 100                   & 6 954.5                                       & 1.450       & 10       & 1.335          & 10          & 1.206             & 10    & \textbf{1.173}    & 10   & 1.290             & 10    & 1.901             & 10    & 1.398              & 10    & 2.634          & 10        \\
                                       &                            & 150                   & 27 398.6                                      & 1.301       & 10       & 1.301          & 10          & 1.164             & 10    & \textbf{1.155}    & 10   & 1.717             & 10    & 5.826             & 10    & 2.873              & 10    & 15.838         & 6         \\ \cline{2-20} 
                                       & \multirow{2}{*}{-0.8}      & 50                    & 19 491.1                                      & 1.606       & 10       & 1.274          & 10          & 1.301             & 10    & 1.272             & 10   & \textbf{1.197}    & 10    & 1.605             & 10    & 1.278              & 10    & 2.165          & 10        \\
                                       &                            & 100                   & 235 384.3                                     & 1.391       & 9        & 1.391          & 9           & 1.231             & 9     & \textbf{1.217}    & 9    & 1.672             & 9     & 3.257             & 9     & 2.235              & 8     & 7.389          & 8         \\ \hline
\multirow{8}{*}{10\textasciicircum{}3} & \multirow{3}{*}{0.8}       & 50                    & 1 255.7                                       & 1.673       & 10       & 1.214          & 10          & 1.221             & 10    & 1.199             & 10   & \textbf{1.155}    & 10    & 1.188             & 10    & 1.195              & 10    & 1.367          & 10        \\
                                       &                            & 100                   & 4 518.5                                       & 1.647       & 10       & 1.245          & 10          & 1.215             & 10    & \textbf{1.172}    & 10   & 1.175             & 10    & 1.232             & 10    & 1.175              & 10    & 1.503          & 10        \\
                                       &                            & 150                   & 7 646.6                                       & 1.597       & 10       & 1.228          & 10          & 1.206             & 10    & 1.180             & 10   & \textbf{1.151}    & 10    & 1.270             & 10    & 1.193              & 10    & 1.678          & 10        \\ \cline{2-20} 
                                       & \multirow{3}{*}{0.0}       & 50                    & 8 662.2                                       & 1.972       & 10       & 1.243          & 10          & 1.407             & 10    & 1.394             & 10   & 1.257             & 10    & 1.242             & 10    & \textbf{1.203}     & 10    & 1.351          & 10        \\
                                       &                            & 100                   & 32 257.5                                      & 1.981       & 10       & 1.282          & 10          & 1.374             & 10    & 1.347             & 10   & 1.262             & 10    & 1.243             & 10    & \textbf{1.208}     & 10    & 1.348          & 10        \\
                                       &                            & 150                   & 59 150.1                                      & 2.051       & 10       & 1.306          & 10          & 1.446             & 10    & 1.399             & 10   & 1.278             & 10    & 1.288             & 10    & \textbf{1.248}     & 10    & 1.368          & 10        \\ \cline{2-20} 
                                       & \multirow{2}{*}{-0.8}      & 50                    & 57 694.4                                      & 1.901       & 10       & 1.318          & 10          & 1.354             & 10    & 1.341             & 10   & 1.239             & 10    & 1.252             & 10    & \textbf{1.230}     & 10    & 1.365          & 10        \\
                                       &                            & 100                   & 239 891.8                                     & 1.921       & 10       & 1.332          & 10          & 1.384             & 10    & 1.355             & 10   & 1.229             & 10    & 1.285             & 10    & \textbf{1.218}     & 10    & 1.421          & 10        \\ \hline
\multirow{8}{*}{10\textasciicircum{}4} & \multirow{3}{*}{0.8}       & 50                    & 1 425.1                                       & 1.749       & 10       & 1.195          & 10          & 1.241             & 10    & 1.252             & 10   & \textbf{1.155}    & 10    & 1.168             & 10    & 1.176              & 10    & 1.430          & 10        \\
                                       &                            & 100                   & 10 520.8                                      & 1.841       & 10       & 1.245          & 10          & 1.305             & 10    & 1.290             & 10   & 1.219             & 10    & \textbf{1.189}    & 10    & 1.201              & 10    & 1.407          & 10        \\
                                       &                            & 150                   & 19 521.4                                      & 1.835       & 10       & 1.211          & 10          & 1.298             & 10    & 1.281             & 10   & 1.193             & 10    & \textbf{1.164}    & 10    & 1.208              & 10    & 1.385          & 10        \\ \cline{2-20} 
                                       & \multirow{3}{*}{0.0}       & 50                    & 11 398.0                                      & 2.025       & 10       & 1.256          & 10          & 1.479             & 10    & 1.465             & 10   & 1.269             & 10    & 1.253             & 10    & \textbf{1.199}     & 10    & 1.359          & 10        \\
                                       &                            & 100                   & 57 365.2                                      & 2.073       & 10       & 1.239          & 10          & 1.454             & 10    & 1.430             & 10   & 1.276             & 10    & 1.231             & 10    & \textbf{1.225}     & 10    & 1.314          & 10        \\
                                       &                            & 150                   & 134 437.2                                     & 2.101       & 10       & 1.260          & 10          & 1.439             & 10    & 1.409             & 10   & 1.264             & 10    & 1.214             & 10    & \textbf{1.212}     & 10    & 1.263          & 10        \\ \cline{2-20} 
                                       & \multirow{2}{*}{-0.8}      & 50                    & 77 565.7                                      & 1.929       & 10       & 1.278          & 10          & 1.403             & 10    & 1.384             & 10   & 1.233             & 10    & 1.227             & 10    & \textbf{1.208}     & 10    & 1.370          & 10        \\
                                       &                            & 100                   & 334 789.5                                     & 2.009       & 10       & 1.273          & 10          & 1.409             & 10    & 1.387             & 10   & 1.249             & 10    & 1.239             & 10    & \textbf{1.215}     & 10    & 1.350          & 10       
\end{tabular}
    \caption{Results for the dynamic strategies.}
    \label{tab:newTabDynamic}
\end{sidewaystable}

The best results for the instance benchmark are attained with dynamic grouping approaches. In instances with $r\geq10^3$, the best-performing dynamic grouping strategies (GGEC2, GGEC3, and GNECU2) can reduce computational cost by up to 18\% compared to their non-grouping counterparts (ECU and SRKB4), and the difference can reach 40\% when compared to F1.
Moreover, GNECU2 improved results obtained with a priori strategies in all instance sets with $r\geq10^3$.

However, in small-range instances, dynamic strategies do not yield favorable results.
For these, GGEC2 is the best-performing strategy among the dynamic grouping strategies. Yet, it is worse than the a priori strategies in $5$ of the $8$ instance sets with $r = 10^2$. In this context, single-triangle dynamic strategies have a clear advantage.
In fact, ECU is the leading strategy over the 8 test sets, with an improvement of $12\%$ in two of them.

The results for single-triangle dynamic strategies show an almost systematic benefit in using extended coverage and ND bound update. ECU provides better average performance than SRKB4 in $21$ out of the $24$ instance sets, with only a minor difference in the remaining three. The most notable gain was $3.75\%$ in the instance set with 
$r = 10^3$, $\rho = 0.0$, and 
$n = 150$. Meanwhile, ECU is at most $1\%$ worse in the $3$ instance sets in which SRKB4 is better.

GNECU2, whose idea is to promote coverage by exploring larger triangles first, is 
the winning strategy. It is best in $10$ out of $24$ test sets and second best in $3$ of those. Moreover, it is at most $4\%$ worse than the best strategy for instances with $r \geq 10^3$. It presents an 
improvement over GNECU3, which did not perform well. Regarding the angle measures, it is also true that GAEC2 is better than GAEC3, despite the latter having decent performance overall.
Such results indicate that
the two-sized variant of each strategy achieves the most powerful balance between grouping and coverage
These results corroborate our observations in Section~\ref{sec:optRes}
confirming that the best performance is achieved with groupings of average size two.
Finally, the dynamic methods that use the ND bound as a measure perform slightly better than those using the group angle, in contrast with the findings for a priori strategies in Section~\ref{sec:AprioriRes}.

\section{Conclusions}
\label{sec:conc}

In this paper, we propose grouping strategies to 
improve the second phase of two-phase methods 
based on ranking algorithms. 
We also present a method for obtaining an optimal grouping for a given instance. 
Although not applicable to large instance sizes, it provides not only a reference to assess the performance of our grouping strategies but also gives further insight into their design. 
For instance, for the Bi-Objective Spanning Tree problem, a group size of more than three is not required.

Our best-performing grouping 
strategies use
information from specific 
group measures of the supported
points identified in the first phase, 
as well as from previously formed groups.
In particular, our extended coverage method, which is based on the work in \cite{steiner2008computing}, gives
a significant improvement.
Experimental results show that these strategies incur only a computational cost of at most 25\% higher than the optimal grouping and improve over the current ranking-based 
two-phase methods.

This paper discusses only two 
measures for guiding group formation: ND bound and group angle. Although not reported here, we also considered other measures, such as the largest distance, or a fraction of it, between the supported point of a group and an upper bound, as well as
the total area covered by the group.
However, these approaches yielded poor
results. It remains an open question whether other geometric measures could further improve our results.

A natural next step is to parallelize our approaches. Although not widely explored in the literature, two-phase methods offer a clear advantage over other approaches for solving bi-objective combinatorial optimization problems, as they are inherently parallelizable. In our case, additional challenges include how to adapt group formation based on the number of processors available and how to reduce communication costs, particularly for dynamic grouping strategies.

Notably, our methods are generalizable to other bi-objective combinatorial optimization problems, with the only requirement being the availability of a ranking algorithm.
If the time budget is too constrained and if optimality is
not a strong requirement, the ranking algorithm may
terminate early.

\section*{Acknowledgments}

This work is partially financed through national funds by FCT - Fundação para a Ciência e a Tecnologia, I.P., in the framework of the Project UIDB/00326/2025 and UIDP/00326/2025. 
The first author acknowledges FCT for the Ph.D. fellowship 2022.14645.BD. 
Part of this work was conducted by the second author as a Visiting Professor at Paris Dauphine University - PSL, and by the first author during a short-term scientific mission at the same institution. This mission was funded by the COST Action Randomised Optimisation Algorithms Research Network (ROAR-NET), CA22137, supported by COST (European Cooperation in Science and Technology).

\bibliographystyle{elsarticle-harv}
\bibliography{main.bib}

\end{document}